\documentclass[12pt,letterpaper,onecolumn,journal]{IEEEtran}

\usepackage[final]{epsfig}
\usepackage{cite}
\usepackage{graphicx}
\usepackage{psfrag}
\usepackage{subfigure}
\usepackage{url}
\usepackage{stfloats}
\usepackage{amsmath, amssymb}
\usepackage[pointedenum]{paralist}
\usepackage{bm}

\usepackage{array}
\DeclareGraphicsExtensions{.eps,.ps}
\graphicspath{{./eps/}}

\title{On Optimum End-to-End Distortion in \\ MIMO Systems}

\author{\IEEEauthorblockN{Jinhui Chen\IEEEauthorrefmark{1}, Dirk T. M. Slock\IEEEauthorrefmark{2}}
\IEEEauthorblockA{\IEEEauthorrefmark{1}Alcatel-Lucent Shanghai Bell \\ Research \& Innovation Center \\388 Ningqiao Rd., Pudong, 201206, Shanghai, China \\ Email: Jinhui.Chen@alcatel-sbell.com.cn\\}
\IEEEauthorblockA{\IEEEauthorrefmark{2}Department of Mobile Communications, EURECOM \\ B.P. 193, 06904, Sophia-Antipolis Cedex, France \\Email: Dirk.Slock@eurecom.fr}
}

\date{}
\begin{document}
\maketitle

\begin{abstract}
This paper presents the joint impact of the numbers of antennas, source-to-channel bandwidth ratio and spatial correlation on the optimum expected end-to-end distortion in an outage-free MIMO system. In particular, based on an analytical expression valid for any SNR, a closed-form expression of the optimum asymptotic expected end-to-end distortion valid for high SNR is derived. It is comprised of the optimum distortion exponent and the multiplicative optimum distortion factor. Demonstrated by the simulation results, the analysis on the joint impact of the optimum distortion exponent and the optimum distortion factor explains the behavior of the optimum expected end-to-end distortion varying with the numbers of antennas, source-to-channel bandwidth ratio and spatial correlation. It is also proved that as the correlation tends to zero, the optimum asymptotic expected end-to-end distortion in the setting of correlated channel approaches that in the setting  of uncorrelated channel. The results in this paper could be performance objectives for analog-source transmission systems. To some extend, they are instructive for system design.
\end{abstract}


\section{Introduction}
\subsection{Background}
\renewcommand{\thefootnote}{}
\footnote{Parts of the work in this paper have been presented in \cite{chen_icc08, chen_globecom08}.} It is well-known that the functional diagram and the basic elements of a digital communication system can be illustrated by Fig.\ref{c1f:chmod}  \cite{proakis}. The source can be either analog (continuous-amplitude) or digital (discrete-amplitude). Whichever is the source, there is always a tradeoff between the efficiency and the reliability. For transmitting a digital sequence, the tradeoff would be between the spectral efficiency (bit/s/Hz) \cite{tse} and the error probability. For transmitting a bandlimited analog source, under the  assumption of a band-limited white Gaussian source, the tradeoff would be between the source-to-channel bandwidth ratio $W_s/W_c$ (SCBR) \cite{goblick_it1965} and the mean squared error (MSE) \cite{shannon48, shannon}, \emph{i.e.}, the end-to-end distortion.

\begin{figure*}
\includegraphics[scale = 0.8]{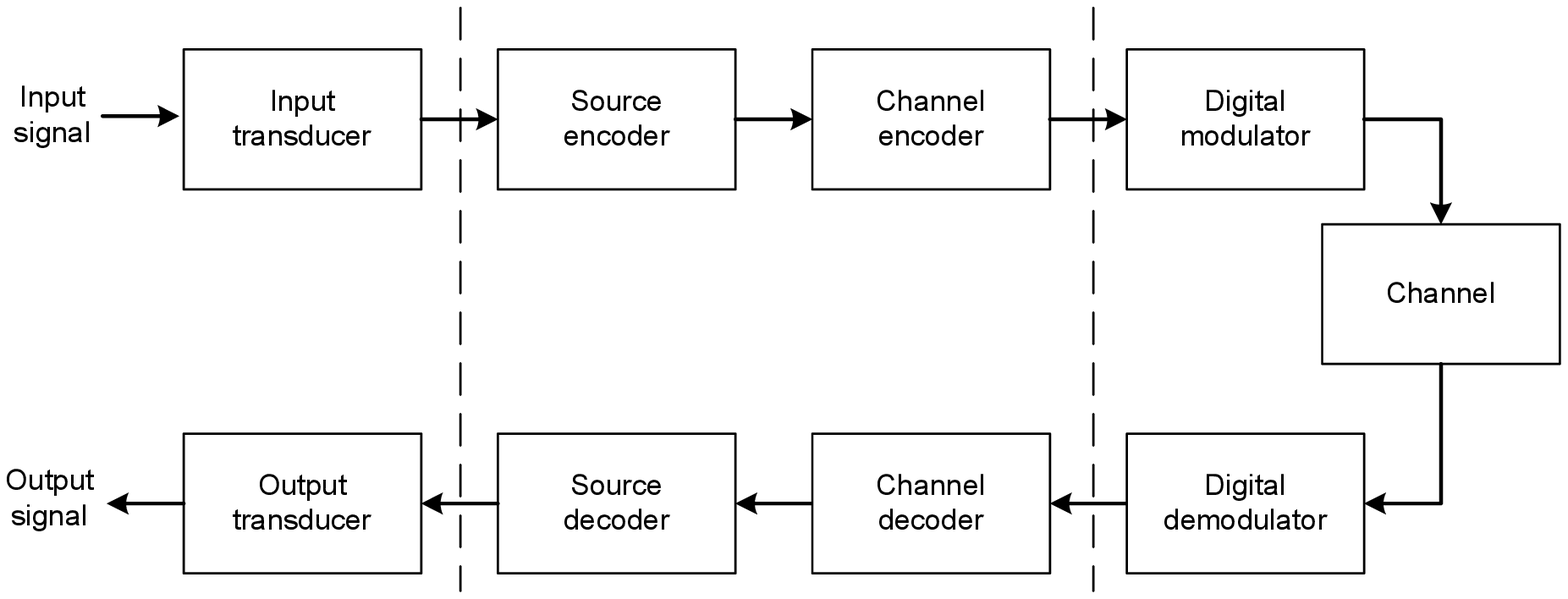}
\centering
\caption{Basic elements of a digital communication system}
\label{c1f:chmod}
\end{figure*}

A point of distinction between digital-source transmission and analog-source transmission is: in digital-source transmission, if the spectral efficiency (bit/s/Hz) is below the upper bound (channel capacity) subject to channel state and the transmitter knows the instantaneous channel state information (CSI) perfectly, the error probability would go to zero; whereas, in analog-source transmission, no matter how good the channel condition and the system are, the end-to-end distortion is non-vanishing, because the entropy of a continuous-amplitude source is infinite and thus the exact recovery of an analog source requires infinite channel capacity \cite{shannon48, shannon, gallager, cover}.

Regarding the end-to-end distortion, in \cite{ziv_it1969, ziv_it1970}, Ziv and Zakai investigated the decay of MSE with SNR for the analog-source transmission over a noisy single-input single-output (SISO) channel without any channel knowledge on the transmitter side (CSIT). In \cite{laneman_isit04, laneman}, Laneman \emph{et al.} used the \emph{distortion exponent} in the asymptotic expected distortion
\begin{equation}
\Delta \triangleq -\lim_{\rho \rightarrow \infty} \frac{ED(\rho)}{\log \rho}
\end{equation}
related to SCBR as a metric to compare different source-channel coding approaches for parallel channels. Note that $\rho$ denotes the SNR and $ED$ denotes the expected end-to-end distortion over all possible channel states. Choudhury and Gibson presented the relations between the end-to-end distortion and the outage capacity for AWGN channels \cite{choudhury_itaw07}. Zoffoli \emph{et al.} studied the characteristics of the distortions in MIMO systems with different strategies, with and without CSIT \cite{zoffoli_globecom08, zoffoli_allerton08}.

In \cite{holliday_allerton04, holliday, holliday_it08}, for tandem source-channel coding systems, assuming optimal block quantization and SNR-dependent rate-adaptive transmission as in \cite{zheng}, Holliday and Goldsmith investigated the expected end-to-end distortion for uncorrelated block-fading MIMO channels based on the results in \cite{zheng, gersho_it79, hochwald_it97}. They gave the following upper bound on the total expected distortion (MSE)
\begin{equation}
ED \leq 2^{-\frac{2r}{\eta}\log\rho+O(1)}+2^{-(N_r-r)(N_t-r)\log\rho+o(\log\rho)}
\label{c1eq:holli1}
\end{equation}
where $\eta$ is the SCBR, $r$ is the multiplexing gain (the source rate scales like $r\log\rho$), $N_t$ is the number of transmit antennas and $N_r$ is the number of receive antennas. Considering the asymptotic high SNR regime,  they proposed that the multiplexing gain $r$ should satisfy
\begin{equation}
\Delta_{\mathrm{sep}}^*  = (N_r-r)(N_t-r) = \frac{2r}{\eta}+o(1)
\label{c1eq:holli2}
\end{equation}
where $\Delta_{\mathrm{sep}}^*$ is the optimum distortion exponent for tandem source-channel coding systems.  The explicit expression of $\Delta_{\mathrm{sep}}^*$ is given by Theorem 2 in \cite{caire_it07},
\begin{equation}
\Delta_{\mathrm{sep}}^* (\eta) = \frac{2\left[jd^*(j-1)-(j-1)d^*(j)\right]}{2+\eta(d^*(j-1)-d^*(j))}, \quad
\eta \in \left[\frac{2(j-1)}{d^*(j-1)}, \frac{2j}{d^*(j)}\right)
\end{equation}
for $j=1, \ldots, N_{\min}$ with $N_{\min} = \min\{N_t, N_r\}$ and $d^*(j) = (N_t-j)(N_r-j)$.
Note that a factor 2 appears here and there because the source is real whereas the channel is complex.

In \cite{caire_allerton05, caire_it07}, assuming an uncorrelated block-fading MIMO channel, perfect CSIT and joint source-channel coding, Caire and Narayanan derived the \emph{optimum distortion exponent}
\begin{equation}
\Delta^*(\eta) = \sum_{i=1}^{N_{\min}}\min\left\{\frac{2}{\eta}, 2i-1+\vert N_t-N_r\vert\right\}
\label{c1eq:optdexp}
\end{equation}
which is larger than $\Delta_{\mathrm{sep}}^*$. Concurrently, the same result as (\ref{c1eq:optdexp}) was also provided by Gunduz and Erkip \cite{gunduz_itw06, gunduz_it08}.

Caire-Narayanan's and Gunduz-Erkip's derivations are extensions to the outage probability analysis in \cite{zheng}. They jointly considered the MIMO-channel mutual information in bits per channel use (bpcu) \cite{telatar}
\begin{equation}
\mathcal{I} = \log\left\vert \mathbf{I}_{N_r\times N_r} + \frac{\rho}{N_t}\mathbf{H}\mathbf{H}^{\dag} \right\vert
\end{equation}
where $\mathbf{H}$ is the $N_r\times N_t$ complex channel matrix with $N_t$ inputs and $N_r$ outputs , the rate-distortion function for a $\mathcal{N}(0,1)$ source \cite{cover}
\begin{equation}
D(R_s) = 2^{-2R_s}
\end{equation}
where $R_s$ is the source rate, and Shannon's rate-capacity inequality for outage-free transmission \cite{shannon}
\begin{equation}
R_s \leq R_c.
\end{equation}

\subsection{Problem statement}

Nevertheless, there is something more than the distortion exponent in the expected end-to-end distortion. Intuitively, for high SNR, the form of the \emph{asymptotic optimum expected end-to-end distortion} can be written as
\begin{equation}
ED^*_{\mathrm{asy}} = \mu^*(\rho)\rho^{-\Delta^*}
\end{equation}
where the multiplicative \emph{optimum distortion factor} $\mu^*(\rho)$ varies less than exponentially:
\begin{equation}
\lim_{\rho\rightarrow \infty} \frac{\log\mu^*(\rho)}{\log \rho} = 0.
\end{equation}

For an analog-source transmission system, its performance at a high SNR could be measured via the asymptotic expected end-to-end distortion
\begin{equation}
ED_{\mathrm{asy}} = \mu(\rho)\rho^{-\Delta}
\end{equation}
where the distortion exponent $\Delta$ and the distortion factor $\mu(\rho)$ could be obtained analytically.
\begin{figure}
\centering
\includegraphics[width = 9cm, height = 8cm]{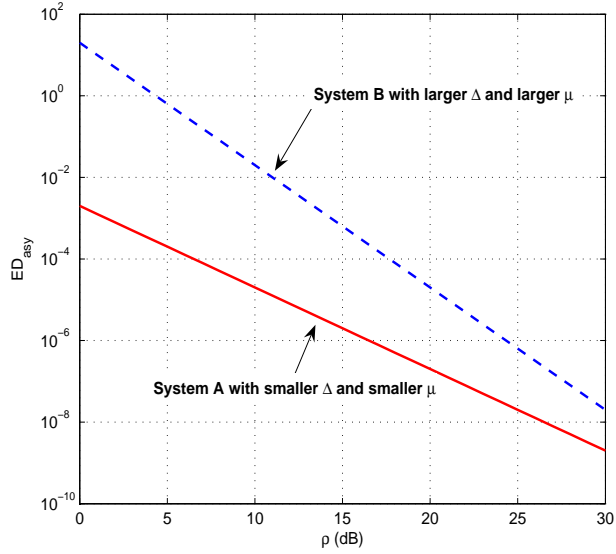}
\caption{Impact of distortion factor}
\label{c1f:egdraw}
\end{figure}

Obviously, we cannot say that a system achieves the optimum asymptotic expected distortion $ED^*_{\mathrm{asy}}$ if what it achieves is only the optimum distortion exponent $\Delta^*$. Also, we cannot say that in the regime of practical high SNR, the scheme with a larger distortion exponent must perform better than the other. As illustrated by Fig.\ref{c1f:egdraw}, in the regime of practical high SNR, the effect of the distortion factor must be taken into consideration. In other words, for practical cases, studying only the optimum distortion exponent is insufficient and giving the closed-form expression of $ED^*_{\mathrm{asy}}$ is more meaningful. Using $ED^*_{\mathrm{asy}}$ as an objective, via analyzing both $\Delta^*$ and $\mu^*(\rho)$, it is possible to design an analog-source transmission system performing better than the existing systems in the regime of practical high SNR.

For deriving $ED^*_{\mathrm{asy}}$, if we could obtain the analytical expression of $ED^*$ valid  for any SNR, then it would be easy to find out the optimum distortion factor $\mu^*(\rho)$ and the optimum distortion exponent $\Delta^*$.

\subsection{Outline}

In this paper, for the cases of spatially uncorrelated channel and correlated channel, we give an analytical expression of the optimum expected end-to-end distortion $ED^*$ in an outage-free MIMO system valid for any SNR, based on which the optimum asymptotic expected end-to-end distortion $ED^*_{\mathrm{asy}}$ is derived. The simulation results agree with our analysis with the derived results on the joint impact of the numbers of antennas, source-to-channel bandwidth ratio and spatial correlation.

The remainder of this paper is organized as follows. The system model is given in Section \ref{sec_model}. In Section \ref{sec_pre}, the preliminaries such as the mathematical definitions, properties and lemmas are presented for deriving the main results in Section \ref{main}. Section \ref{num} is dedicated to the simulation results, numerical analysis, and discussions. Finally, the contributions of this paper are concluded in Section \ref{sec_con}, with our perspectives on future work.

Throughout the paper, vectors and matrices are denoted by bold characters, $\vert \mathbf{A} \vert$ denotes the determinant of matrix $\mathbf{A}$ and $\{a_{ij}\}_{i,j = 1,..., N}$ is an $N\times N$ matrix with entries $a_{ij}$, $i,j = 1, ..., N$. Also, $\mathbb{E}\{\cdot\}$ denotes expectation and, in particular, $\mathbb{E}_x\{\cdot\}$ denotes expectation over the random variable $x$. The superscript $^\dag$ denotes conjugate transpose. $(a)_n$ denotes $\Gamma(a+n)/\Gamma(a)$. $\log$ refers to the logarithm with base 2.

\section{MIMO System Model}\label{sec_model}

Assume that a continuous-time white Gaussian source $s(t)$ of bandwidth $W_s$ and source power $P_s$ is to be transmitted over a flat block-fading MIMO channel of bandwidth $W_c$ and the system is working on ``short" frames due to strict time delay constraint, \emph{i.e.}, no time diversity can be exploited. The transmission system is supposed to be free of outage, \emph{e.g.}, the transmitter knows the instantaneous channel capacity by scalar feedback and does joint source-channel coding. Let $\hat{s}(t)$ denote the recovered source at the receiver.

Suppose a $K$-to-$(N_t\times T)$ joint source-channel encoder is employed at the transmitter \cite{caire_it07}, which maps the source block $\mathbf{s}^{'} \in \mathbb{R}^K$ onto channel codewords $\mathbf{X} \in \mathbb{C}^{N_t\times T}$. Herein, the source block $\mathbf{s}^{'}$ is composed of $K$ source samples, $N_t$ is the number of transmit antennas, and $T$ is the number of channel uses for transmitting one block. The corresponding source-channel decoder is a mapping $\mathbb{C}^{N_r\times T} \rightarrow \mathbb{R}^K$ that maps the channel output $\mathbf{Y} = \{\mathbf{y}_1, \ldots, \mathbf{y}_T\}$ into an approximation $\hat{\mathbf{s}}^{'}$. Assuming the continuous-time source $s(t)$ is sampled by a Nyquist sampler, $2W_s$ samples per second, and the bandlimited MIMO channel is used as a discrete-time channel at $2W_c$ channel uses per second \cite[pp.\,247-250]{cover}, we have the SCBR
\begin{equation}
\eta = \frac{W_s}{W_c} = \frac{K}{T}.
\end{equation}

At the $t^{\mathrm{th}}$ channel use, the output of the discrete-time flat block-fading MIMO channel with $N_t$ inputs and $N_r$ outputs is
\begin{equation}
\mathbf{y}_t = \mathbf{H}\mathbf{x}_t + \mathbf{n}_t, \quad t = 1,\ldots, T
\end{equation}
where $\mathbf{x}_t \in \mathbb{C}^{N_t}$ is the transmitted signal satisfying the long-term power constraint $\mathbb{E}[\mathbf{x}_t^H\mathbf{x}_t] = P$, $\mathbf{H} \in \mathbb{C}^{N_r\times N_t}$ is the channel matrix whose entries $h_{ij} \sim  \mathcal{CN}(0,1)$, $\mathbf{n}_t \in \mathbb{C}^{N_r}$ is the additive white noise matrix whose entries $n_{t,i} \sim \mathcal{CN}(0,\sigma_n^2)$. Note that the SNR per receive antenna is $\rho = P/\sigma_n^2$.

In the case of uncorrelated channel, the $h_{ij}$'s are independent to each other. In the case of  spatially correlated channel, we have the correlation matrix $\bm{\Sigma} = \mathbb{E}(\mathbf{H}\mathbf{H}^{\dag})$ which is assumed to be a full-rank matrix with distinct  eigenvalues $\bm{\sigma} = \{\sigma_1, \sigma_2, \cdots, \sigma_{N_{\min}}\}$, $0 < \sigma_1 < \sigma_2 < \cdots <\sigma_{N_{\min}}$. It can be seen that in the case of uncorrelated channel, $\bm{\Sigma}$ is an identity matrix with $\sigma_1 = \sigma_2 = \cdots = \sigma_{N_{\min}} = 1$.

\section{Mathematical Preliminaries}\label{sec_pre}

The mathematical properties, definitions and lemmas in this section will be used in the derivations for the main results.

\subsection{Mathematical properties and definitions}

We shall use the integral of an exponential function
\begin{equation}
\begin{split}
&\int_0^\infty e^{-px}x^{q-1}(1+ax)^{-\nu}dx = a^{-q}\Gamma(q)\Psi(q, q+1-\nu, p/a), \\
& \quad \Re\{q\} > 0, \quad \Re\{p\} > 0, \quad \Re\{a\} > 0.\\
\end{split}
\label{eq_int}
\end{equation}
as introduced in \cite[pp.\,365]{tables}. This involves the confluent hypergeometric function
\begin{equation}
\Psi(a,c; x) = \frac{1}{\Gamma(a)}\int_0^{\infty}e^{-xt}t^{a-1}(1+t)^{c-a-1}dt, \quad \Re\{a\}>0
\end{equation}
which satisfies (with $y=\Psi$)
\begin{equation}
x\frac{d^2y}{dx^2}+(c-x)\frac{dy}{dx}-ay = 0.
\end{equation}
Bateman has given a thorough analysis on $\Psi(a,c;x)$ \cite[pp.\, 257-261]{bateman}. In particular, he obtained the expressions on $\Psi(a, c; x)$ for small $x$ as Table \ref{tb_psisx} shows. In Appendix \ref{app1}, we also state some of his more general results for any $x$, which we will use for the analysis in the case of spatially correlated MIMO channel.
\begin{table}
\caption{$\Psi(a, c; x)$ for small $x$, real $c$}
\label{tb_psisx}
\centering
\begin{tabular}{|c|c|}
\hline
$\mathbf{c}$ & $\mathbf{\Psi}$ \\ \hline
$c>1$ & $x^{1-c}\Gamma(c-1)/\Gamma(a)+ o\left(x^{1-c}\right)$ \\ \hline
$c=1$ & $-\left[\Gamma(a)\right]^{-1}\log x + o\left(|\log x|\right)$ \\ \hline
$c<1$ & $\Gamma(1-c)/\Gamma(a-c+1) + o(1)$ \\ \hline
\end{tabular}
\end{table}

\subsection{Mathematical lemmas}
The proofs of the mathematical lemmas below can be found in Appendices \ref{appl1}-\ref{appl7}.

\newtheorem{lemma}{Lemma}
\begin{lemma}\label{l1}
Define an $m\times m$ full-rank matrix $\mathbf{W}(x)$ whose $(i,j)^{\mathrm{th}}$ entry is of the  form $c_{ij}x^{\mathrm{min}\{a, i+j\}}$, $c_{ij} \neq 0$, $x, a \in \mathbb{R}^+$, $1 \leqslant i,j \leqslant m $. Then
\begin{equation}
\lim_{x\rightarrow 0}\frac{\mathrm{log}\vert\mathbf{W}(x)\vert}{\mathrm{log}x} = \sum_{i=1}^m\mathrm{min}\{a, 2i\}.
\end{equation}
\end{lemma}

\begin{lemma}\label{l2}
Define an $m\times m$ Hankel matrix $\mathbf{W}(x)$ whose $(i,j)^{\mathrm{th}}$ entry is of the form $c_{i+j}x^{i+j}$, $c_{i+j} \neq 0$, $x\in \mathbb{R}^+$, $1 \leqslant i,j \leqslant m $. Then, each summand in the determinant of $\mathbf{W}(x)$ has the same degree $m(m+1)$ over $x$.
\end{lemma}

\begin{lemma}\label{l3}
Define an $m\times m$ Hankel matrix $\mathbf{W}$ whose $(i,j)^{\mathrm{th}}$ entry is $\Gamma(a+i+j-1)$, $1 \leqslant i,j \leqslant m $, $a\in \mathbb{R}$. Then
\begin{equation}
\vert\mathbf{W}\vert = \prod_{k=1}^m \Gamma(k)\Gamma(a+k).
\end{equation}
\end{lemma}

\begin{lemma}\label{l4}
Define an $m\times m$ Hankel matrix $\mathbf{W}$ whose $(i,j)^{\mathrm{th}}$ entry is $\Gamma(a+i+j-1)\Gamma(b-i-j+1)$ where $1 \leqslant i,j \leqslant m $, $m \geqslant 2$ and $a, b \in \mathbb{R}$. Then
\begin{equation}
\begin{split}
\vert\mathbf{W}\vert  &= \Gamma(a+1)\Gamma(b-1)\Gamma^{m-1}(a+b)\\
& \quad \times \prod_{k=2}^m  \Gamma(k)\Gamma(a+k)\frac{\Gamma(b-2k+2)\Gamma(b-2k+1)}{\Gamma(a+b-k+1)\Gamma(b-k+1)}.
\end{split}
\end{equation}
\end{lemma}

\begin{lemma}\label{l8}
Define an $m\times m$ Toeplitz matrix $\mathbf{W}$ whose $(i,j)^{\mathrm{th}}$ entry is $\Gamma(a+i-j)$, $1 \leqslant i,j \leqslant m $, $a\in \mathbb{R}$. Then
\begin{equation}
\vert\mathbf{W}\vert = (-1)^{\frac{m(m-1)}{2}}\prod_{k=1}^m \Gamma(k)\Gamma(a+k-m).
\end{equation}
\end{lemma}

\begin{lemma}\label{l6}
Define
\begin{align}
f(n) &= \prod_{k=1}^m \frac{\Gamma(n-m-a+k)}{\Gamma(n-k+1)}, \\
g(n) &= n^{am}f(n), \end{align}
subject to $a \in \mathbb{R}^+$, $m,n \in \mathbb{Z}^+$, $n \geq m$, and $ n-m+1 \geq a$.
Then both $f(n)$ and $g(n)$ are monotonically decreasing.
\end{lemma}

\begin{lemma}\label{l7}
Let $(a)_n$ denote $\Gamma(a+n)/\Gamma(a)$, $a\in \mathbb{R}$,  $n \in \mathbb{Z}^+$. Then
\begin{equation}
(a+1)_n = (-1)^n(-a-n)_n.
\end{equation}
\end{lemma}

\section{Main Results} \label{main}

\subsection{Uncorrelated MIMO channel}

\newtheorem{theorem}{Theorem}
\begin{theorem}[Optimum Expected Distortion over an Uncorrelated MIMO Channel]
Assume a continuous-time white Gaussian source $s(t)$ of bandwidth $W_s$ and power $P_s$ to be transmitted over an uncorrelated block-fading MIMO channel of bandwidth $W_c$.  The optimum expected end-to-end distortion is
\begin{equation}
ED_{\mathrm{unc}}^{*}(\eta) = \frac{P_s\vert\mathbf{U}(\eta)\vert}{\prod_{k=1}^{N_{\mathrm{min}}}\Gamma(N_\mathrm{max} -k+1)\Gamma(N_{\mathrm{min}}-k+1)}
\label{Dlb}
\end{equation}
where $\eta = W_s/W_c$ (SCBR), $N_{\mathrm{min}} = \min\{N_t, N_r\}$, $N_{\mathrm{max}} = \max\{N_t, N_r\}$, and  $\mathbf{U}(\eta)$ is an $N_{\mathrm{min}}\times N_{\mathrm{min}}$ Hankel matrix whose $(i,j)^{\mathrm{th}}$ entry is
\begin{equation}
\begin{split}
u_{ij}(\eta) = \left(\frac{\rho}{N_t}\right)^{-d_{ij}}\Gamma(d_{ij}) \Psi\left(d_{ij}, d_{ij}+1-\frac{2}{\eta}; \frac{N_t}{\rho}\right)
\label{uij}
\end{split}
\end{equation}
where $d_{ij} = i+j+\vert N_t-N_r\vert-1$, $1 \leq i, j \leq N_{\min}$, and $\Psi(a, b; x)$ is the $\Psi$ function (see \cite[pp.\,257-261]{bateman}). This theorem is valid for any SNR.
\label{t1}
\end{theorem}
\begin{proof}
The source rate of the source $s(t)$ is \cite{shannon48}
\begin{equation}
R_s = W_s\log \frac{P_s}{D}
\end{equation}
where $D$ is the distortion (MSE).

Under the assumption that the transmitter only knows the instantaneous channel capacity $R_c$, the covariance matrix of the transmitted vector $\mathbf{x}$ at the transmitter is taken to be a scaled identity matrix $P/N_t\cdot\mathbf{I}_{N_t}$. As stated in  \cite{telatar}, the mutual information per MIMO channel use is
\begin{equation}
\mathcal{I}(\mathbf{x};\mathbf{y}) = \log\left\vert\mathbf{I}_{N_r}+\frac{\rho}{N_t}\mathbf{H}\mathbf{H}^\dag\right\vert.
\end{equation}
And as stated in \cite[pp.\,248-250]{cover}, a channel of bandwidth $W_c$ can be represented by samples taken $1/2W_c$ seconds apart, i.e., the channel is used at $2W_c$ channel uses per second as a time-discrete channel. Hence, the channel capacity (bit/second) is
\begin{equation}
R_c = 2 W_c\mathcal{I} = 2W_c \log\left\vert\mathbf{I}_{N_r}+\frac{\rho}{N_t}\mathbf{H}\mathbf{H}^\dag\right\vert. \label{Rc}
\end{equation}
Substituting (\ref{Rc}) into Shannon's rate-capacity inequality
\begin{equation}
R_s \leq R_c,
\end{equation}
we get the optimum end-to-end distortion
\begin{equation}
D^*(\eta) = P_s\left\vert\mathbf{I}_{N_r}+ \frac{\rho}{N_t}\mathbf{H}\mathbf{H}^\dag\right\vert^{-\frac{2}{\eta}}.
\end{equation}

Thereby, the optimum expected end-to-end distortion is
\begin{equation}
ED^*(\eta) = P_s\mathbb{E}_{\mathbf{H}}\left\vert\mathbf{I}_{N_r}+ \frac{\rho}{N_t}\mathbf{H}\mathbf{H}^\dag\right\vert^{-\frac{2}{\eta}}
\label{EDopt},
\end{equation}
whose form is analogous to the moment generating function of capacity in \cite{chiani}. By the mathematical results given by Chiani \emph{et al.}  \cite{chiani} for the expectation over an uncorrelated MIMO Gaussian channel $\mathbf{H}$, we have
\begin{equation}
ED^*_{\mathrm{unc}}(\eta) = P_sK\vert\mathbf{U(\eta)}\vert
\end{equation}
where $\mathbf{U}(\eta)$ is an $N_{\min}\times N_{\min}$ Hankel matrix with $(i,j)^{\mathrm{th}}$ entry given by
\begin{equation}
u_{ij}(\eta) = \int_0^{\infty}x^{N_{\max}-N_{\min}+j+i-2}e^{-x} \left(1+\frac{\rho}{N_t}x\right)^{-\frac{2}{\eta}}dx
\label{uij2}
\end{equation}
and
\begin{equation}
K = \frac{1}{\prod_{k=1}^{N_{\mathrm{min}}}\Gamma(N_\mathrm{max} -k+1)\Gamma(N_{\mathrm{min}}-k+1)}.
\end{equation}
By the integral solution (\ref{eq_int}), (\ref{uij2}) can be written in the analytic form
\begin{equation}
u_{ij}(\eta) = \left(\frac{\rho}{N_t}\right)^{-d_{ij}}\Gamma(d_{ij}) \Psi\left(d_{ij}, d_{ij}+1-\frac{2}{\eta}; \frac{N_t}{\rho}\right),
\end{equation}
This concludes the proof of the theorem.
\end{proof}
$\\$

Theorem \ref{t1} tells us that the analytical expression of $ED^*_{\mathrm{unc}}$ is a polynomial in $\rho^{-1}$. Therefore, for high SNR, the optimum asymptotic expected end-to-end distortion is of the form
\begin{equation}
ED^*_{\mathrm{asy,unc}} = \mu^*_{\mathrm{unc}}(\eta)\rho^{-\Delta_{\mathrm{unc}}^*(\eta)}
\end{equation}
where $\Delta_{\mathrm{unc}}^*(\eta)$ is the \emph{optimum distortion exponent} satisfying
\begin{equation}
\Delta^*_{\mathrm{unc}}(\eta) = -\lim_{\rho\rightarrow\infty} \frac{\log ED^*_{\mathrm{unc}}(\eta)}{\log\rho}
\end{equation}
and $\mu^*_{\mathrm{unc}}$ is the accompanying \emph{optimum distortion factor} satisfying
\begin{equation}
\lim_{\rho\rightarrow\infty}\frac{\log\mu^*_{\mathrm{unc}}(\eta)}{\log\rho}=0.
\end{equation}
Since $ED^*_{\mathrm{unc}}$ is concave in the log-log scale and monotonically decreasing with SNR and  $ED^*_{\mathrm{asy,unc}}$ is the tangent of the curve $ED^*_{\mathrm{unc}}$ at the point where SNR is infinitely high, we see that the asymptotic tangent line $ED^*_{\mathrm{asy,unc}}$ is always above the curve $ED^*_{\mathrm{unc}}$, \emph{i.e.}, $ED^*_{\mathrm{asy,unc}}$ is always worse than $ED^*_{\mathrm{unc}}$.

The closed-form expressions of  $\Delta_{\mathrm{unc}}^*(\eta)$ and  $\mu^*_{\mathrm{unc}}(\eta)$ are given as follows.
\\

\begin{theorem}[Optimum Distortion Exponent over an Uncorrelated MIMO Channel]The optimum distortion exponent is
\begin{equation}
\Delta^*_{\mathrm{unc}}(\eta) = \sum_{k=1}^{N_{\min}}\min\left\{\frac{2}{\eta}, 2k-1+\vert N_t-N_r \vert\right\}.
\label{c2eq:t2}
\end{equation}
\label{t2}
\end{theorem}
\begin{proof}
This optimum distortion exponent appeared already in \cite{caire_it07, gunduz_itw06}. However, a different proof is provided here.

Consider $u_{ij}(\eta)$ in Theorem \ref{t1}. When $\rho$ is large, $N_t/\rho$ is small. We thus refer to Table \ref{tb_psisx} and see that, for high SNR, $u_{ij}(\eta)$ approaches  $e_{ij}(\eta)\rho^{-\Delta_{ij}(\eta)}$ with
\begin{equation}
\Delta_{ij}(\eta)  = \min\left\{\frac{2}{\eta},  i + j -1+\vert N_t-N_r \vert \right\}
\end{equation}
and
\begin{equation}
\lim_{\rho\rightarrow\infty}\frac{\log e_{ij}(\eta)}{\log \rho} = 0.
\end{equation}
Straightforwardly, in the regime of high SNR, the asymptotic form of $\vert\mathbf{U}(\eta)\vert$ can be represented by $\vert\mathbf{E}(\eta)\vert\rho^{-\Delta^*_{\mathrm{unc}}(\eta)}$
with
\begin{equation}
\lim_{\rho\rightarrow\infty}\frac{\log\vert\mathbf{E}(\eta)\vert}{\log\rho} = 0.
\end{equation}
By Lemma \ref{l1}, we obtain that
\begin{equation}
\Delta^*_{\mathrm{unc}}(\eta) = \sum_{k=1}^{N_{\min}}\min\left\{\frac{2}{\eta}, 2k-1+\vert N_t-N_r \vert\right\}.
\end{equation}
This concludes the proof of this theorem.
\end{proof}

\begin{figure*}[!t]
\hrulefill
\begin{align}
\small
\kappa_l(\beta, t, m, n) &=
\begin{cases}
\Gamma(n-m+1)\frac{\Gamma(\beta-n+m-1)}{\Gamma(\beta)}\prod_{k=2}^t \Gamma(k)\Gamma(n-m+k) \\
\quad \times \frac{\Gamma(\beta-n+m-2k+2)\Gamma(\beta-n+m-2k+1)} {\Gamma(\beta-k+1)\Gamma(\beta-n+m-k+1)}, \quad & t > 1;\\ \Gamma(n-m+1)\frac{\Gamma(\beta-n+m-1)}{\Gamma(\beta)}, \quad & t = 1; \\
1, \quad & t=0.\label{kl}
\end{cases}\\
\kappa_h(\beta, t, m, n) &=
\begin{cases}
\prod_{k=1}^t\Gamma(k)\Gamma(n-m-\beta+k), \qquad & t > 0;\\
1, \quad & t = 0. \label{kh}
\end{cases}
\end{align}
\hrulefill
\end{figure*}

\normalsize
\begin{theorem}[Optimum Distortion Factor over an Uncorrelated MIMO Channel]
Define two four-tuple functions $\kappa_l(\beta, t, m, n)$ and $\kappa_h(\beta, t, m, n)$ for $\beta \in \mathbb{R}^+$ and $t \in \{0, \mathbb{Z}^+\}$ as in (\ref{kl}) and (\ref{kh}).
The optimum distortion factor $\mu^*_{\mathrm{unc}}(\eta)$ is given as follows:

\begin{compactenum}
\item For $2/\eta \in (0, \vert N_t-N_r\vert+1)$, referred to as \emph{the high SCBR regime} (HSCBR), the optimum distortion factor is
\begin{equation}
\mu^*_{\mathrm{unc}}(\eta) = P_s{N_t}^{\Delta^*_{\mathrm{unc}}}\;\frac{\kappa_h(\frac{2}{\eta}, N_{\min}, N_{\min}, N_{\max} )}{\prod_{k=1}^{N_{\min}}\Gamma(N_{\max}-k+1)\Gamma(N_{\min}-k+1)}.
\label{chi}
\end{equation}
It decreases monotonically  with $N_{\max}$.
\item For $2/\eta \in (N_t+N_r-1, +\infty)$, referred to as \emph{the low SCBR regime} (LSCBR), the optimum distortion factor is
\begin{equation}
\mu^*_{\mathrm{unc}}(\eta) = P_s{N_t}^{\Delta^*_{\mathrm{unc}}}\;\frac{\kappa_l(\frac{2}{\eta}, N_{\min}, N_{\min}, N_{\max} )}{\prod_{k=1}^{N_{\min}}\Gamma(N_{\max}-k+1)\Gamma(N_{\min}-k+1)}.
\end{equation}
\item For $2/\eta \in [\vert N_t-N_r\vert+1, N_t+N_r-1]$, referred to as \emph{the moderate SCBR regime} (MSCBR), the optimum distortion factor is
\begin{equation}
\mu^*_{\mathrm{unc}}(\eta) = \begin{cases} P_s{N_t}^{\Delta^*_{\mathrm{unc}}}\;\frac{\kappa_l(\frac{2}{\eta}, l, N_{\min}, N_{\max})\kappa_h(\frac{2}{\eta}-2l, N_{\min}-l, N_{\min}, N_{\max})}{\prod_{k=1}^{N_{\min}}\Gamma(N_{\max}-k+1)\Gamma(N_{\min}-k+1)}, \\
\quad \mod\{\frac{2}{\eta}+1-\vert N_t-N_r\vert,2\} \neq 0; \\
P_s{N_t}^{\Delta^*_{\mathrm{unc}}}\log\rho\;\frac{\kappa_l(\frac{2}{\eta}, l-1, N_{\min}, N_{\max})\kappa_h(\frac{2}{\eta}-2l, N_{\min}-l, N_{\min}, N_{\max})} {\prod_{k=1}^{N_{\min}}\Gamma(N_{\max}-k+1)\Gamma(N_{\min}-k+1)},\\
\quad \mod\{\frac{2}{\eta}+1-\vert N_t-N_r\vert,2\} = 0
\end{cases}
\end{equation}
where $l = \left\lfloor \frac{\frac{2}{\eta}+1-\vert N_t-N_r\vert}{2}\right\rfloor$.
\end{compactenum} \label{t3}
\end{theorem}
\begin{proof}
See Appendix \ref{appunc}.
\end{proof}

\subsection{Spatially correlated MIMO channel}
\begin{theorem}[Optimum Expected Distortion over a Correlated MIMO Channel]
The optimum expected end-to-end distortion in a system over a spatially correlated MIMO channel is
\begin{equation}
ED_{\mathrm{cor}}^{*}(\eta) = \frac{P_s\vert\mathbf{G}(\eta)\vert} {\prod_{k=1}^{N_{\mathrm{min}}} \sigma_k^{\vert N_t-N_r\vert+1}\Gamma(N_{\max} -k+1)\prod_{1\leq m < n \leq N_{\min}}(\sigma_n-\sigma_m)}.
\label{Dlb_cor}
\end{equation}
where $\mathbf{G}(\eta)$ is an $N_{\mathrm{min}}\times N_{\mathrm{min}}$ matrix whose $(i,j)^{\mathrm{th}}$ entry given by
\begin{equation}
\begin{split}
g_{ij}(\eta) = \left(\frac{\rho}{N_t}\right)^{-d_{j}}\Gamma(d_{j}) \Psi\left(d_{j}, d_{j}+1-\frac{2}{\eta}; \frac{N_t}{\sigma_i\rho}\right).
\label{gij}
\end{split}
\end{equation}
$d_j= \vert N_t-N_r\vert+j$. $\bm{\sigma} = \{\sigma_1, \sigma_2, \cdots, \sigma_{N_{\min}}\}$ with $0 < \sigma_1 < \sigma_2 < \cdots <\sigma_{N_{\min}}$ denoting the ordered eigenvalues of the correlation matrix $\bm{\Sigma}$.
\label{t4}
\end{theorem}

\begin{proof}
Following the proof of Theorem \ref{t1}, by the mathematical results given by Chiani \emph{et al.} in \cite{chiani} for a spatially correlated $\mathbf{H}$, we have
\begin{equation}
ED_{\mathrm{cor}}^{*}(\eta) = P_sK_{\bm{\Sigma}}\vert\mathbf{G}(\eta)\vert
\end{equation}
where $\mathbf{G}(\eta)$ is an $N_{\min} \times N_{\min}$ matrix with $(i,j)^{\mathrm{th}}$ entry given by
\begin{equation}
g_{ij}(\eta) = \int_0^\infty x^{\vert N_t-N_r\vert +j-1}e^{-x/\sigma_i}(1+\frac{\rho}{N_t}x)^{-\frac{2}{\eta}} dx
\label{gijint}
\end{equation}
and
\begin{equation}
K_{\bm{\Sigma}} = \frac{|\bm{\Sigma}|^{-N_{\max}}} {\vert\mathbf{V}_2(\bm{\sigma})\vert\prod_{k=1}^{N_{\min}}\Gamma\left(N_{\max}-k+1\right)}
\end{equation}
where $\mathbf{V}_2(\sigma)$ is a Vandermonde matrix given by
\begin{equation}
\mathbf{V}_2(\bm{\sigma}) \triangleq \mathbf{V}_1 \left(-\{\sigma_1^{-1}, \cdots, \sigma_{N_{\min}}^{-1} \}\right) \label{c3_defv2}
\end{equation}
with the Vandermonde matrix $\mathbf{V}_1(\mathbf{x})$ defined as
\begin{equation}
\mathbf{V}_1(\mathbf{x}) \triangleq \left( \begin{array}{cccc}
1 & 1 & \cdots & 1 \\
x_1 & x_2 & \cdots & x_{N_{\min}} \\
\vdots & \vdots & \ddots & \vdots \\
x_1^{N_{\min}-1} & x_2^{N_{\min}-1} & \cdots & x_{N_{\min}}^{N_{\min}-1}
\end{array} \right).\label{c3_defv1}
\end{equation}
In terms of the property of a Vandermonde matrix \cite{horn}, the determinant of $\mathbf{V}_2(\bm{\sigma})$
\begin{align}
\vert \mathbf{V}_2(\bm{\sigma}) \vert &= \prod_{1\leq m < n \leq  N_{min}}(-\sigma_j^{-1}+\sigma_i^{-1})\\
&= \prod_{1\leq m < n \leq N_{min}}\sigma_m^{-1}\sigma_n^{-1}(\sigma_n-\sigma_m) \\
&= \prod_{k=1}^{N_{\min}}\sigma_k^{1-N_{\min}}\prod_{1\leq m < n \leq N_{min}}(\sigma_n-\sigma_m)\\
&= \prod_{k=1}^{N_{\min}}\sigma_k^{1-N_{\min}}\vert \mathbf{V}_1(\bm{\sigma})\vert.
\end{align}
Thereby,
\begin{equation}
K_{\bm{\Sigma}} = \frac{1} {\prod_{k=1}^{N_{\mathrm{min}}} \sigma_k^{\vert N_t-N_r\vert +1}\Gamma(N_{\max} -k+1)\prod_{1\leq m < n \leq N_{\min}}(\sigma_n-\sigma_m)}\; .
\end{equation}
In terms of the integral solution (\ref{eq_int}), (\ref{gijint}) can be written in the analytic form
\begin{equation}
g_{ij}(\eta) = \left(\frac{\rho}{N_t}\right)^{-d_j}\Gamma(d_j)\Psi\left(d_j, d_j+1-\frac{2}{\eta}; \frac{N_t}{\sigma_i\rho}\right).
\end{equation}
This concludes the proof of this theorem.
\end{proof}

\begin{theorem}[Optimum Distortion Exponent over a Correlated MIMO Channel]
The optimum distortion exponent  $\Delta^*_{\mathrm{cor}}$ in the case of spatially correlated MIMO channel is the same as the optimum distortion exponent $\Delta^*_{\mathrm{unc}}$ in the case of uncorrelated MIMO channel, that is,
\begin{equation}
\Delta^*_{\mathrm{cor}}(\eta) = \Delta^*_{\mathrm{unc}}(\eta) =
\sum_{k=1}^{N_{\min}}\min\left\{\frac{2}{\eta}, 2k-1+\vert N_t-N_r\vert\right\}
\end{equation}
\label{t5}
\end{theorem}
\begin{proof}
See Appendix \ref{appt5}.
\end{proof}

\begin{theorem}[Optimum Distortion Factor over a Correlated MIMO Channel]\label{t6}
The optimum distortion factor $\mu_{\mathrm{cor}}^*(\eta)$ is given as follows.
\begin{enumerate}
\item For $2/\eta \in (0, \vert N_t-N_r\vert +1)$ (HSCBR), the optimum distortion factor is
\begin{equation}
\mu_{\mathrm{cor}}^*(\eta) =\prod_{k=1}^{N_{\min}}\sigma_k^{-\frac{2}{\eta}}\, \mu^*_{\mathrm{unc}}(\eta).
\end{equation}
\item For $2/\eta \in (N_t + N_r -1, +\infty)$ (LSCBR), the optimum distortion factor is
\begin{equation}
\mu_{\mathrm{cor}}^*(\eta) = \prod_{k=1}^{N_{\min}}\sigma_k^{-N_{\max}}\, \mu^*_{\mathrm{unc}}(\eta).
\end{equation}
\item For $2/\eta \in [\vert N_t - N_r\vert +1, N_t+N_r -1]$ (MSCBR), the optimum distortion factor is
\begin{equation}
\begin{split}
\mu_{\mathrm{cor}}^*(\eta) &=
\frac{(-1)^{\frac{l(l-1)}{2}}\vert\mathbf{V}_3(\bm{\sigma})\vert} {\prod_{k=1}^{N_{\min}}\sigma_k^{\vert N_t-N_r\vert+1}\prod_{1\leq m < n \leq N_{\min}}(\sigma_n-\sigma_m)} \\
&\quad \times\prod_{k=1}^{N_{\min}-l}\frac{(k)_l}{(\vert N_t-N_r\vert-\frac{2}{\eta}+l+k)_l}\,\mu^*_{\mathrm{unc}}(\eta)
\end{split}
\end{equation}
where $l = \lfloor \frac{\frac{\eta}{2}+1-\vert N_r-N_t\vert}{2}\rfloor$ and each entry of $\mathbf{V}_3(\bm{\sigma})$ is
\begin{equation}
v_{3,ij}= \sigma_i^{-\min\{j-1, \frac{2}{\eta}-d_j\}}.
\end{equation}
\end{enumerate}
\end{theorem}
\begin{proof}
See Appendix \ref{appt6}.
\end{proof}

\begin{theorem}[Convergence] \label{t7}
\begin{equation}
\lim_{\bm{\Sigma}\rightarrow \mathbf{I}}\mu^*_{\mathrm{cor}}(\eta) = \mu^*_{\mathrm{unc}}(\eta).
\end{equation}
\end{theorem}
\begin{proof}
See Appendix \ref{appt7}.
\end{proof}

\section{Numerical Analysis and Discussion} \label{num}

In this section, the examples in various settings are provided. The simulation and numerical results illustrate the foregoing results.

\subsection{An example in the HSCBR regime, uncorrelated MIMO channel}

Fig.\ref{fig1} shows the numerical and simulation results on the optimum expected end-to-end distortions in the outage-free MIMO systems over uncorrelated block-fading MIMO channels in the high SCBR regime and at the high SNR, $\rho = 30$ dB. The number of antennas on one side (either the transmitter side or the receiver side) is fixed to five and the number of antennas on the other side is varying. $ED_{\mathrm{unc,sim}}^*$ denotes the $ED^*_{\mathrm{unc}}$ corresponding to (\ref{EDopt}), evaluated by {10 000} realizations of $\mathbf{H}$.

From Fig.\ref{fig1b}, we see that $ED^*_{\mathrm{unc,sim}}$ monotonically decreases with the number of antennas on one side, which agrees with our intuition. There is an excellent agreement between $ED_{\mathrm{unc, asy}}^*$ and $ED_{\mathrm{unc,sim}}^*$, which indicates that, in the setting when SNR is 30 dB,  the behavior of $ED^*_{\mathrm{unc}}$ at a high SNR can be explained by studying $ED_{\mathrm{unc, asy}}^*$.

In Fig.\ref{fig1a}, in terms of Theorem \ref{t2}, the optimum distortion exponent $\Delta_{\mathrm{unc}}^*$ increases with $N_{\min}$ and then remains constant when $N_{\min}$ stops increasing, though the number of antennas on one side is increasing. In Fig.\ref{fig1b}, in terms of Theorem \ref{t3}, $\mu_{\mathrm{unc}}^*$ is monotonically decreasing with $N_{\max}$. Therefore, when $N_{\min} \leq 5$, $ED^*_{\mathrm{unc}}$ is decreasing because $\Delta_{\mathrm{unc}}^*$ is increasing; although the optimum distortion factor $\mu^*_{\mathrm{unc}}$ is increasing, the increase of $\Delta_{\mathrm{unc}}^*$ dominates the tendency of $ED_{\mathrm{unc}}^*$ since the SNR is high. When the $N_{\min}$ is fixed to 5, $ED_{\mathrm{unc}}^*$ is decreasing because $\mu_{\mathrm{unc}}^*$ is decreasing, though $\Delta_{\mathrm{unc}}^*$ keeps constant. In summary, we see that, for high SNR, the decrease of $ED^*_{\mathrm{unc}}$ with the number of antennas is due to either the increase of the optimum distortion exponent or the decrease of the optimum distortion factor.
\begin{figure}
\subfigure[ ]{
\includegraphics[width = 7.5cm , height = 7cm]{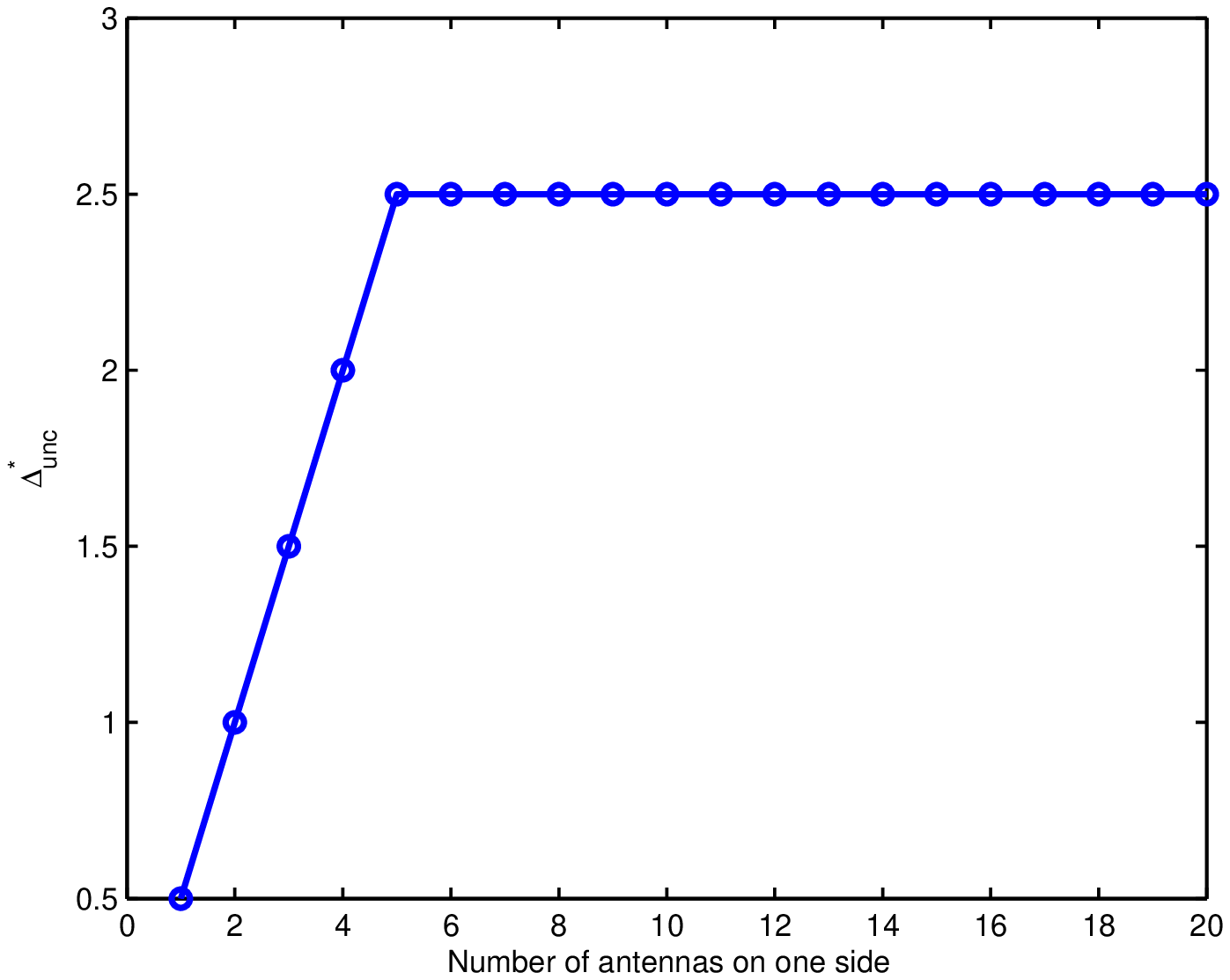}
\label{fig1a}
}
\subfigure[ ]{
\includegraphics[width = 7.5cm , height = 7cm]{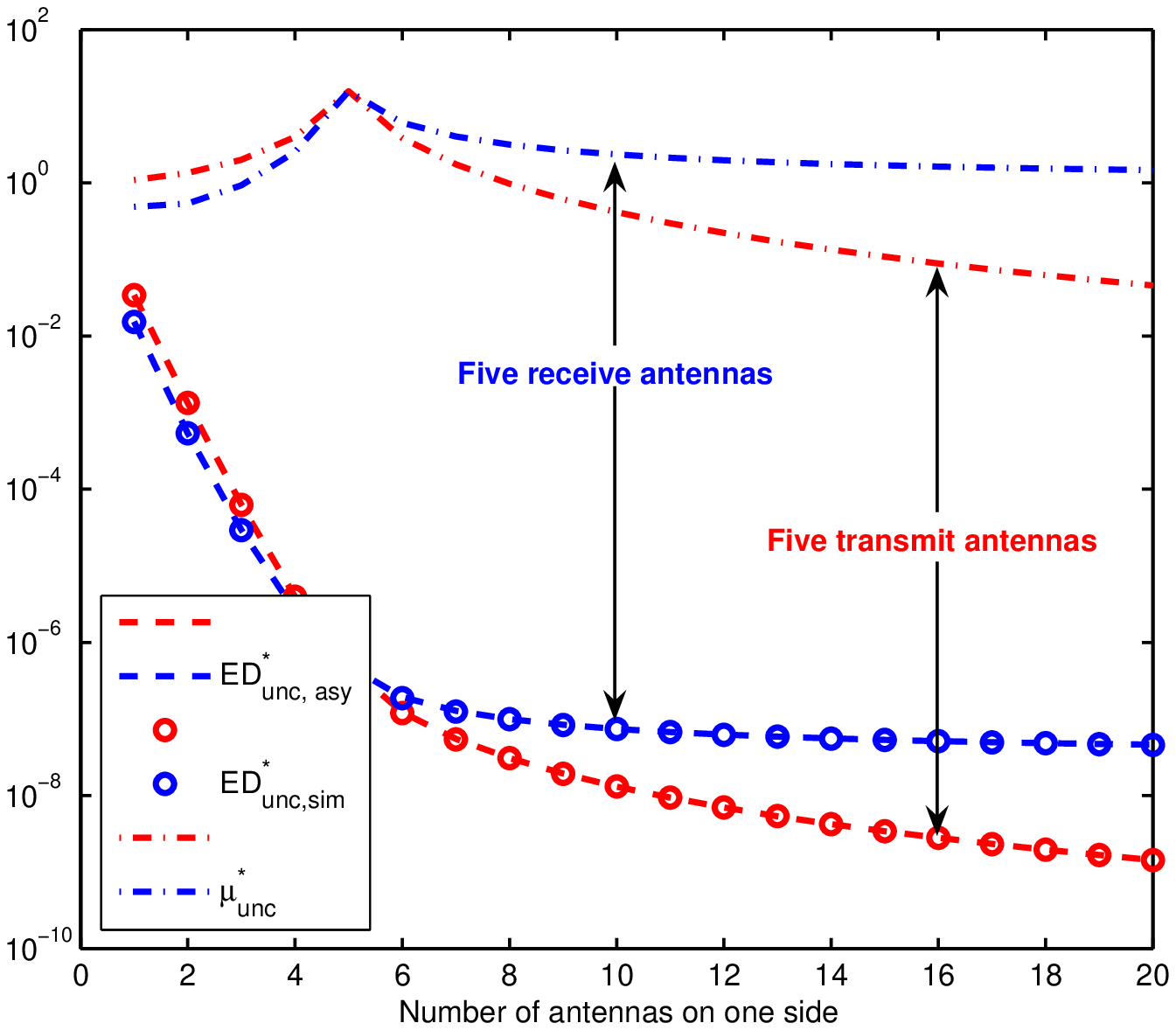}
\label{fig1b}}
\caption{Uncorrelated channel, one of $(N_t, N_r)$ is fixed to 5, $\eta =4$, high SCBR.}
\label{fig1}
\end{figure}

Moreover, from Fig.\ref{fig1}, it is seen that the commutation between the number of transmit antennas and the number of receive antennas impacts $ED^*_{\mathrm{unc}}$. This impact comes from the effect on the optimum distortion factor $\mu_{\mathrm{unc}}^*$. As indicated by the expressions in Theorem \ref{t3} and shown in Fig.\ref{fig1b}, between a couple of commutative antenna allocation schemes, $(N_t=N_{\min}, N_r = N_{\max})$ and $(N_t = N_{\max}, N_r = N_{\min})$, the former scheme whose number of transmit antennas is the smaller between the two numbers of antennas suffers less distortion than the other. This is reasonable since under a certain total transmit power constraint, the scheme with fewer transmit antennas achieves higher average transmit power per transmit antenna.

\subsection{An example in the MSCBR regime, uncorrelated MIMO channel}

In \cite{zoffoli_globecom08, zoffoli_allerton08}, assuming a $\mathcal{N}(0,1)$ source and the system bandwidth is normalized to unity, Zoffoli \emph{et al.} studied the characteristics of the distortions in $2\times 2$ MIMO systems with different space-time coding strategies. In particular, in \cite{zoffoli_allerton08}, assuming the transmitter knows the instantaneous channel capacity and thus the system is free of outage, they compared the strategies with respect to expected distortion and the cumulative density function of distortion. They exhibited that, among REP (repetition), ALM (Alamouti) and SM (spatial multiplexing) strategies, the expected distortion of the ALM strategy is very close to that of the SM strategy.

As Zoffoli et al. derived \cite{zoffoli_allerton08}, the expected distortion of the ALM strategy is
\begin{equation}
ED_{\mathrm{ALM}} = \frac{2}{3}\cdot\frac{\rho\left[(\rho-4)\rho-4\right]+ 4e^{\frac{2}{\rho}}(3\rho+2)\Gamma(0, \frac{2}{\rho})}{\rho^5} \label{edalm}
\end{equation}
and the expected distortion of the SM strategy is
\begin{equation}
ED_{\mathrm{SM}} = -\frac{16\left[\rho-(\rho+2)e^{\frac{2}{\rho}}\Gamma(0,\frac{2}{\rho})\right]^2} {\rho^6}+\frac{8\left[\rho-2e^{\frac{2}{\rho}}\Gamma(0,\frac{2}{\rho})\right] \left[\rho(\rho+2)-4(\rho+1)e^{\frac{2}{\rho}}\Gamma(0,\frac{2}{\rho})\right]}{\rho^6}. \label{edsm}
\end{equation}
Note that $\Gamma(a,x)$ denotes the upper incomplete gamma function, $\Gamma(a,x) = \int_x^{\infty}t^{a-1}e^{-t}dt$. As given in \cite{zoffoli_allerton08}, Fig.\ref{fig2a} shows the difference between the expected distortions of the two strategies in log-lin scale. In log-lin scale, the expected distortion of the ALM strategy is very close to that of the SM strategy in the high SNR regime, i.e., $ED_{\mathrm{ALM}} - ED_{\mathrm{SM}}$ is very small.

\begin{figure}
\subfigure[Log-lin scale]{
\includegraphics[width = 7.5cm , height = 7cm]{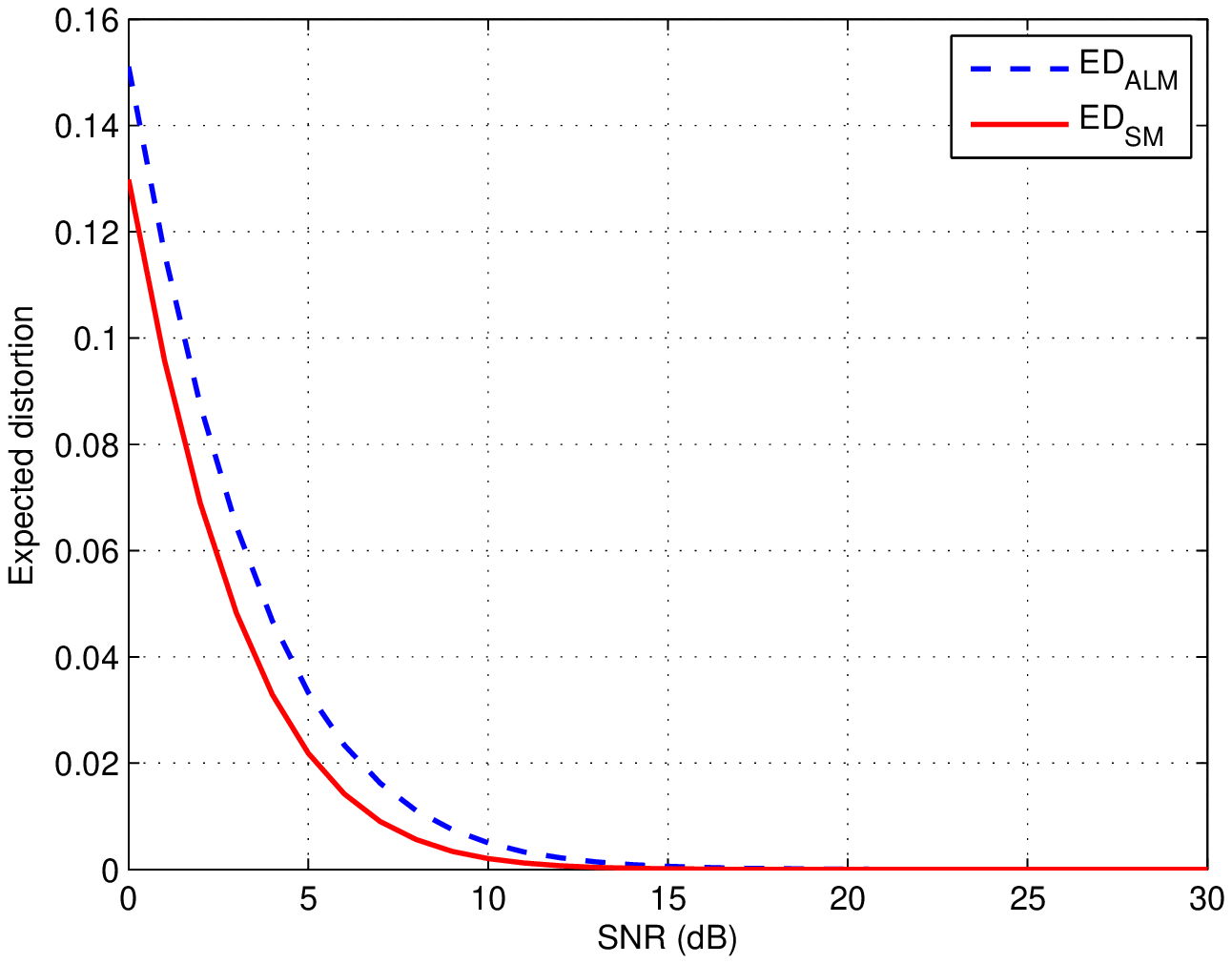}
\label{fig2a}
}
\subfigure[Log-log scale]{
\includegraphics[width = 7.5cm , height = 7cm]{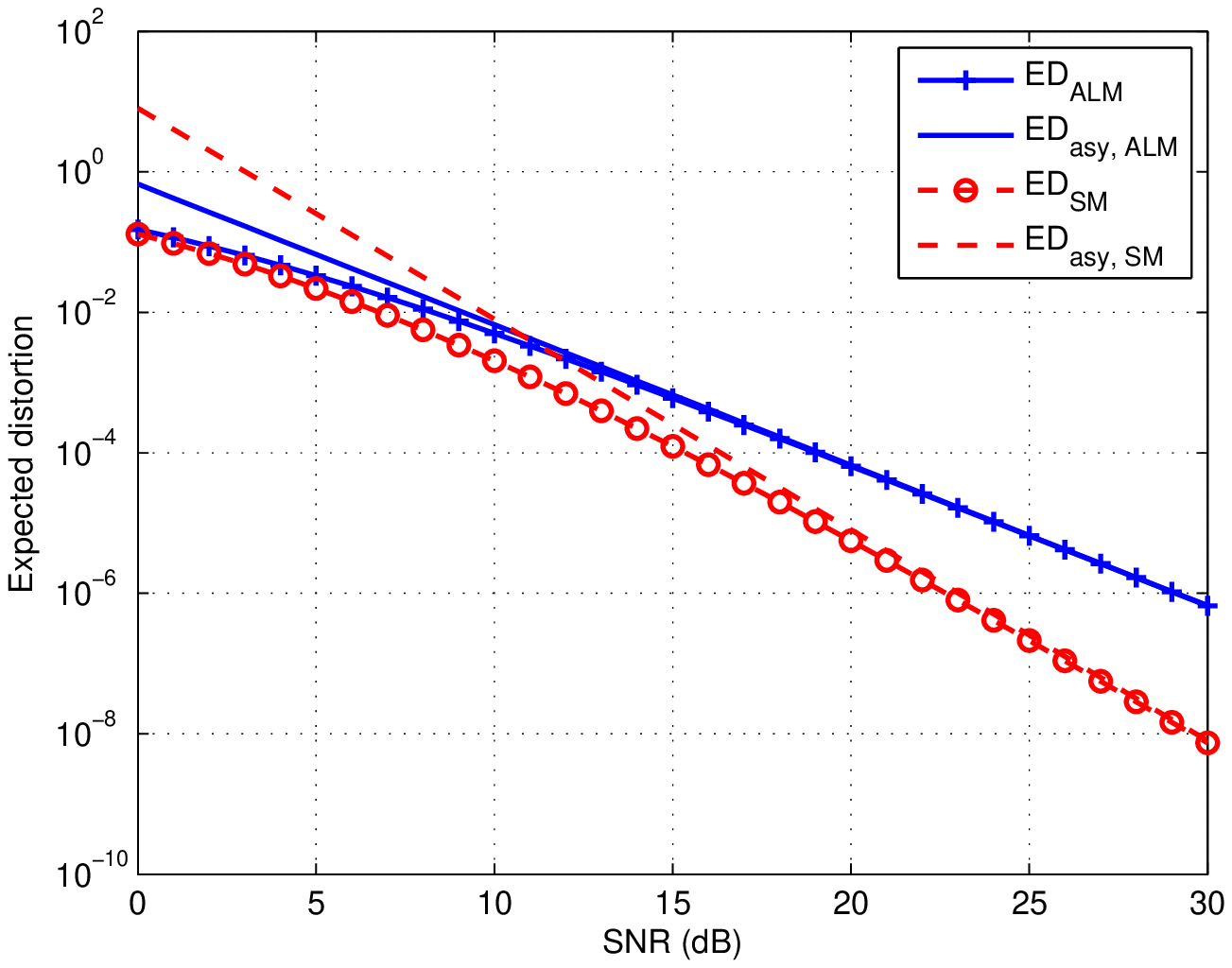}
\label{fig2b}}
\caption{ALM vs. SM, uncorrelated channel, $N_t = N_r = 2$, $\eta =1$, moderate SCBR.}
\label{fig2}
\end{figure}

According to the assumption in \cite{zoffoli_allerton08}, the SCBR of the systems is one, i.e., $\eta = 1$. As $N_t = N_r = 2$, it is seen that, for the systems considered,
\begin{equation}
\vert N_t-N_r\vert +1 < \frac{2}{\eta } < N_t+N_r-1
\end{equation}
and thus the systems are in the moderate SCBR regime.  From the description of SM strategy, it is seen that the expected distortion achieved by SM strategy is the optimum expected distortion for a $2\times 2$ MIMO system with $\eta = 1$, \emph{i.e.}, $ED_{\mathrm{SM}} = ED^*_{\mathrm{unc}}$. Regarding the asymptotic characteristics, from (\ref{edalm}) and (\ref{edsm}), we have
\begin{align}
ED_{\mathrm{asy, ALM}}= \frac{2}{3}\rho^{-2},\\
ED_{\mathrm{asy, SM}} = ED^*_{\mathrm{asy, unc}} = 8\rho^{-3}.
\end{align}

The ratio $ED_{\mathrm{ALM}}/ED_{\mathrm{SM}}$ is an alternative metric revealing the difference between $ED_{\mathrm{ALM}}$ and $ED_{\mathrm{SM}}$, illustrated by Fig.\ref{fig2b} in log-log scale. We see that in the high SNR regime, although $ED_{\mathrm{ALM}}$ approaches $ED_{\mathrm{SM}}$ in the linear scale as Fig.\ref{fig2a} shows, the ratio $ED_{\mathrm{ALM}}/ED_{\mathrm{SM}}$ becomes larger and larger as  Fig.\ref{fig2b} shows. It can also be seen that the expected distortions of the ALM and SM strategies are determined by their asymptotic expressions when the SNR's are greater than 13 dB and 20 dB respectively.

\subsection{An example in the LSCBR regime, uncorrelated MIMO channel}\label{eguncl}

\begin{figure}
\centering
\includegraphics[width = 7.5cm , height = 7cm]{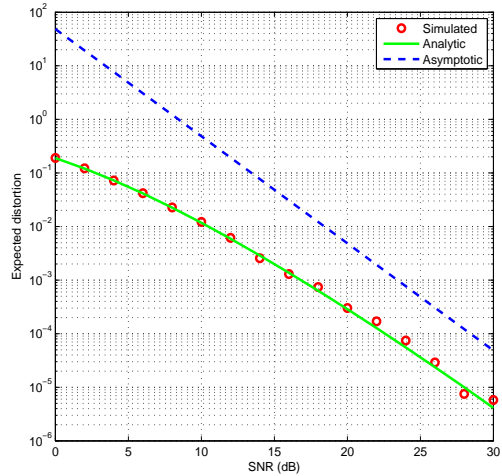}
\caption{Uncorrelated channel, $N_t = 1$, $N_r = 2$, $\eta = 0.99$, low SCBR}
\label{fig3}
\end{figure}

Fig. \ref{fig3} presents an example when $N_t = 1$, $N_r = 2$ and $\eta = 0.99$. The red circles represent the results of Monte Carlo simulations which are carried out by generating 10~000 realizations of $\mathbf{H}$ and evaluating (\ref{EDopt}). The blue dashed line represents $ED^*_{\mathrm{asy,unc}}$. The green line represents the analytical expression of $ED^*_{\mathrm{unc}}$ in Theorem \ref{t1}. It can be seen that the simulated results agree well with our analytical results. The gap between the asymptotic tangent line and the curve of $ED^*_{\mathrm{unc}}$ implies that, for the systems in the LSCBR regime, more terms in the polynomial of $ED^*_{\mathrm{unc}}$ are to be analyzed, which is much more complicated than analyzing the asymptotic expression. It is a subject for future research.

\subsection{Examples in HSCBR \& LSCBR regimes, spatially correlated MIMO channel}

The analytical framework we derived is general and valid for all correlated cases with distinct (unrepeated) eigenvalues of the correlation matrix $\bm{\Sigma}$. To give an example, we consider a well-known correlation model as in \cite{chiani}: the exponential correlation with $\bm{\Sigma}= \{r^{\vert i-j\vert}\}_{i,j = 1,\cdots, N_r}$ and $r \in (0,1)$ \cite{aalo}.

\begin{figure}
\subfigure[$N_t = 4$, $N_r = 2$, $\eta = 10$, high SCBR]{\label{fig4a}
\includegraphics[width = 7.5cm , height = 7cm]{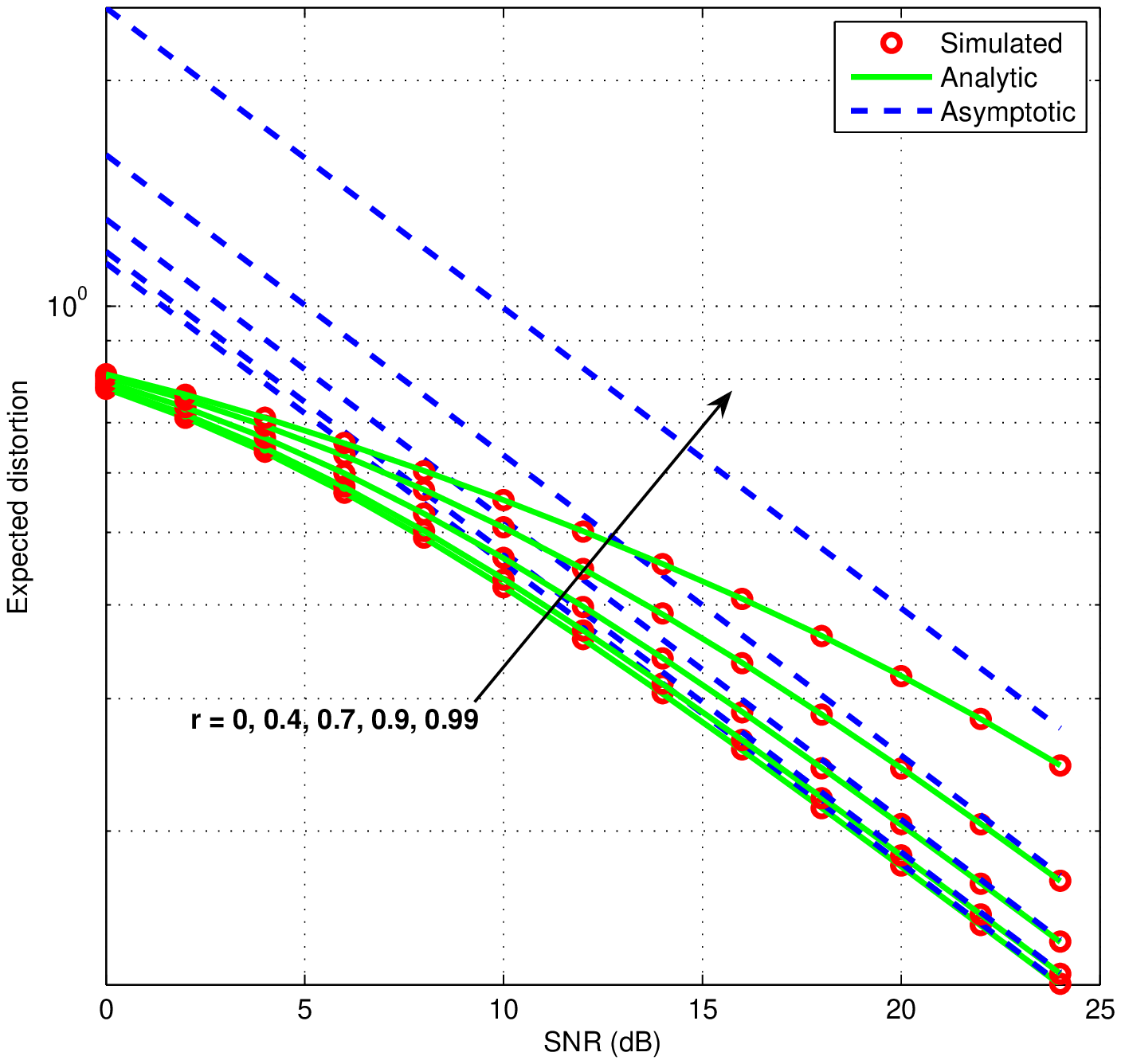}}
\subfigure[$N_t = 2$, $N_r = 2$, $\eta = 0.6657$, low SCBR]{\label{fig4b}
\includegraphics[width = 7.5cm , height = 7cm]{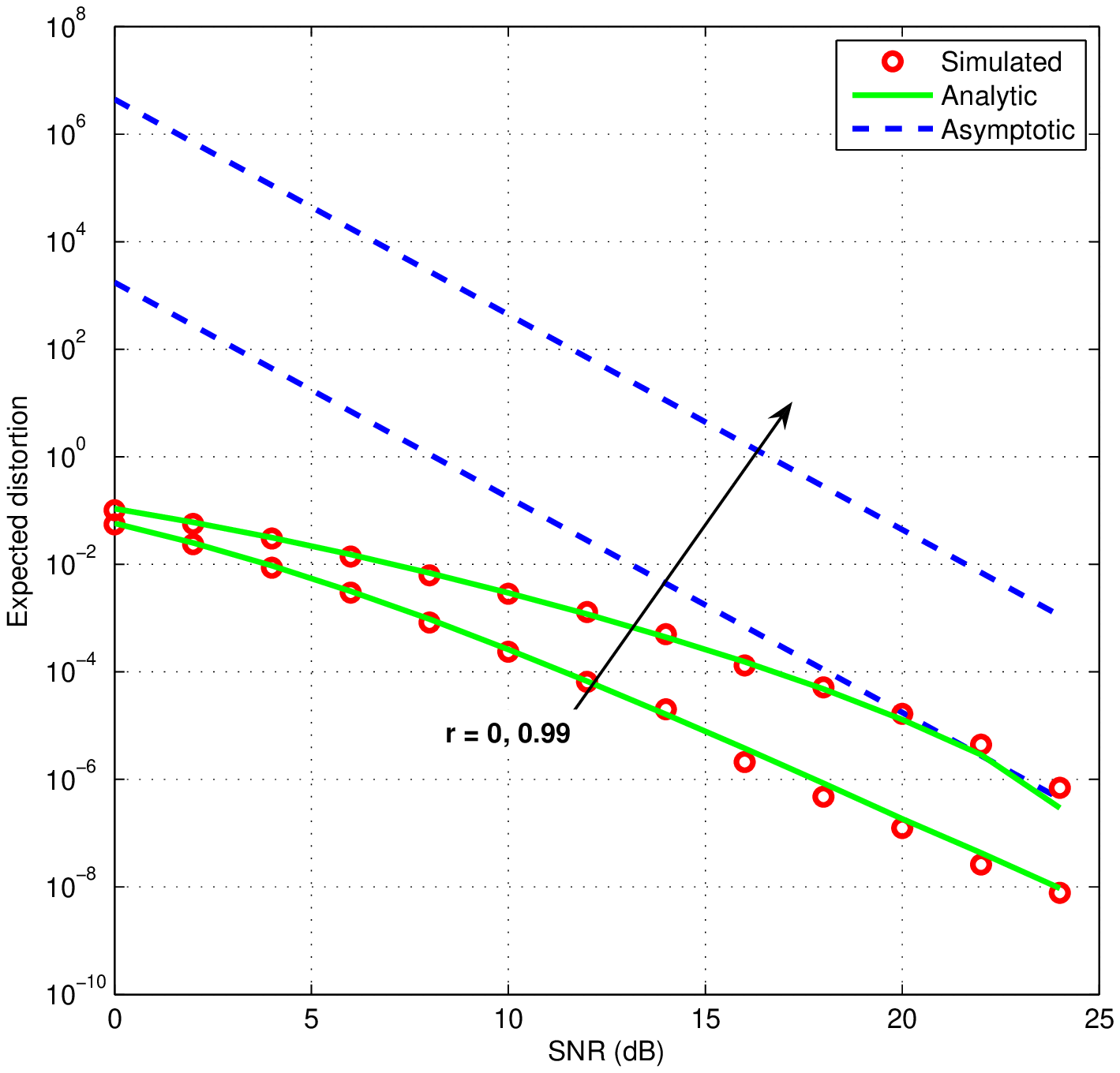}}
\caption{Expected distortions of uncorrelated and correlated channels}
\label{fig4}
\end{figure}

Fig.\ref{fig4} illustrates the optimum expected end-to-end distortion $ED^*$ on a power-one white Gaussian source transmitted in different correlation scenarios. Red circles represent the results of Monte Carlo simulations which are carried out by generating 10 000 realizations of $\mathbf{H}$ and evaluating (\ref{EDopt}). Green lines represent the analytical expressions of $ED^*_{\mathrm{cor}}$ in Theorem \ref{t4} and $ED^*_{\mathrm{unc}}$ in Theorem 1. Blue dashed lines represent the optimum asymptotic expected end-to-end distortion $ED^*_{\mathrm{asy}}$.
\begin{equation}
ED^*_{\mathrm{asy}} = \begin{cases}
\mu^*_{\mathrm{unc}}\rho^{-\Delta^*_{\mathrm{unc}}}, \quad r=0 \\
\mu^*_{\mathrm{cor}}\rho^{-\Delta^*_{\mathrm{cor}}}, \quad r>0.
\end{cases}
\end{equation}

In Fig.\ref{fig4a}, we see that there is an agreement between $ED^*$ and $ED^*_{\mathrm{asy}}$ in the high SNR regime. Corresponding to Theorem \ref{t5} and Theorem \ref{t6}, in the high SNR regime, due to the same optimum SNR distortion exponent, the optimum distortions of the systems in different correlation scenarios have the same descendent slopes; the difference comes from different distortion factors which depend on the correlation coefficients. The optimum distortion is increasing with $r$ and the line of the uncorrelated case ($r=0$) is the lowest. For reaching the same optimum expected distortion, there is about 8 dB difference of SNR between the cases of $r = 0.99$  and $r = 0$. This agrees with our intuition that spatial correlation decreases channel capacity.

The impact of correlation can also be seen in Fig.\ref{fig4b} by the example in the low SCBR regime. There are gaps between the asymptotic lines and the optimum expected distortions for the same reason as for the example in Section \ref{eguncl}.

\section{Conclusion and Future Work}\label{sec_con}

\subsection{Conclusion}

In this paper, considering transmitting a white Gaussian source $s(t)$ over a MIMO channel in an outage-free system, we have derived the analytical expression of the optimum expected end-to-end distortion valid for any SNR (see Theorem \ref{t1} and Theorem \ref{t4}) and the closed-form asymptotic expression of the optimum asymptotic expected end-to-end distortion (see Theorem \ref{t2}, Theorem \ref{t3}, Theorem \ref{t5} and Theorem \ref{t6}) comprised of the optimum distortion exponent and the multiplicative optimum distortion factor. By the results on the optimum asymptotic expected end-to-end distortion, we have analyzed the joint impact of the numbers of antennas, source-to-channel bandwidth ratio (SCBR) and spatial correlation on the optimum expected end-to-end distortion. Straightforwardly, our results are bounds for outage-bearing systems and could be the performance objectives for analog-source transmission systems. To some extend, they are instructive for system design.

\subsection{Future work}

\begin{itemize}
\item
As we have shown in Fig.\ref{fig3} and Fig.\ref{fig4b}, for a system in the low SCBR regime, there is an apparent gap between $ED_{\mathrm{asy}}^*$ and $ED^*$ in the practical high SNR regime. The reason that the gap exists is the effect of the other terms in the polynomial expansion of $ED^*$. Therefore, if the closed-form expression with more terms in the polynomial expansion of $ED^*$ could be derived, the analysis on the behavior of $ED^*$ would be more precise.
\item
Let us provide an insight into Theorem \ref{t2}. Define a non-negative integer $m$ as
\begin{equation}
m = \begin{cases}
N_{\min},& \quad  0 < \frac{2}{\eta} < \vert N_t - N_r\vert+1; \\
N_{\min} - \left\lfloor \frac{\frac{2}{\eta}+1-\vert N_t-N_r\vert}{2}\right\rfloor, & \quad \vert N_t - N_r\vert+1 \leq \frac{2}{\eta} \leq N_t + N_r-1; \\
0, &\quad \frac{2}{\eta} > N_t + N_r-1.
\end{cases}
\end{equation}
Then, (\ref{c2eq:t2}) can be written in the form
\begin{equation}
\Delta^*(\eta) = (N_t - m)(N_r - m) + \frac{2m}{\eta},
\label{c2eq:dexp2}
\end{equation}
which looks analogous to the formula of the Diversity-Multiplexing Tradeoff (DMT) \cite{zheng} and to the expression of the distortion exponent (\ref{c1eq:holli2}) in tandem source-channel coding systems \cite{holliday_it08}. Note that (\ref{c2eq:dexp2}) has nothing to do with outage since the instantaneous channel capacity is assumed to be known at the transmitter. This intriguing similarity induces us to conjecture that there may be a hidden connection to be explored.
\end{itemize}

\section{Acknowledgements}
The authors would like to thank Prof. Giussepe Caire in USC for kindly providing his important manuscript not-yet-published at that time and Prof. Emre Teletar in EPFL for his detailed review and suggestions on the first author's dissertation including this work. Special thanks to Dr. Junbo Huang for the inspiring discussions on mathematic-relevant derivations and borrowing the Bateman's book published in 1953 from INRIA's library for the first author, which became the mathematical basis of this work.

Eurecom's research is partially supported by its industrial members: Swisscom, Thales, SFR, Orange, STEricsson, SAP, BMW Group, Cisco, Monaco Telecom, and Symantec. The research work leading to this paper has also been partially supported by the European Commission under the ICT research network of excellence NEWCOM++ of the 7th Framework programme and by the French ANR-RNRT project APOGEE.

\appendices
\section{Some Properties of $\Psi(a,c;x)$}\label{app1}

\begin{itemize}
\item If $c$ is not an integer,
\begin{equation}
\begin{split}
\Psi(a,c;x) &= \frac{\Gamma(1-c)}{\Gamma(a-c+1)}\Phi(a,c;x)\\
&\quad + \frac{\Gamma(c-1)}{\Gamma(a)}x^{1-c}\Phi(a-c+1,2-c;x)
\end{split}
\label{eq_psicni}
\end{equation}
where $\Phi(a,c;x)$ is another confluent hypergeometric function,
\begin{equation}
\Phi(a,c;x) = \sum_{r=0}^{\infty}\frac{(a)_r}{(c)_r}\frac{x^r}{r!}.
\end{equation}
Note that $(a)_n = \Gamma(a+n)/\Gamma(a)$.

\item if $c$ is a positive integer,
\small{
\begin{equation}
\begin{split}
&\Psi(a,n+1;x) = \frac{(-1)^{n-1}}{n!\Gamma(a-n)}\Big{[}\Phi(a,n+1;x)\log x \\
& \quad + \sum_{r=0}^{\infty}\frac{(a)_r}{(n+1)_r}\left(\psi(a+r)-\psi(1+r)-\psi(1+n+r)\right)\frac{x^r}{r!}\Big{]}\\
& \quad + \frac{(n-1)!}{\Gamma(a)}\sum_{r=0}^{n-1}\frac{(a-n)_r}{(1-n)_r}\frac{x^{r-n}}{r!} \quad n = 0,1,2,...
\end{split}
\label{eq_psicpi}
\end{equation}
}
\normalsize
The last sum is to be omitted if $n=0$.

\item
\begin{equation}
\Psi(a,c;x) = x^{1-c}\Psi(a-c+1, 2-c; x).
\label{eq_psieq}
\end{equation}
Thus, when $c$ is a non-positive integer, we can  obtain the form of  $\Psi(a,c;x)$ from (\ref{eq_psicpi}) and (\ref{eq_psieq}),
\small{
\begin{equation}
\begin{split}
&\Psi(a,c;x) = \frac{(-1)^{-c}}{(1-c)!\Gamma(a)}\Big{[}\Phi(a+1-c,2-c;x)x^{1-c}\log x \\
& \quad + \sum_{r=0}^{\infty}\frac{(a+1-c)_r}{(2-c)_r} \big(\psi(a+1-c+r)-\psi(1+r) \\
& \quad -\psi(2-c+r)\big) \frac{x^{r+1-c}}{r!}\Big{]}+ \frac{\Gamma(1-c)}{\Gamma(a+1-c)}\sum_{r=0}^{-c}\frac{(a)_r}{(c)_r}\frac{x^{r}}{r!}
\end{split}
\label{eq_psicnpi}
\end{equation}}
\end{itemize}

\section{Proof of Lemma \ref{l1}} \label{appl1}

We will prove this lemma recursively.

Define $p(n) = \mathrm{min}\{a, n\}$, subject to  $a \in \mathbb{R}^+$ and  $n \in \mathbb{Z}^+$. If $m_1 - m_2 = n_1 - n_2$, $m_1 > n_1$, and $m_2 > n_2$, then
\begin{equation}
p(m_1)-p(m_2) \leq p(n_1) - p(n_2).\label{p1}
\end{equation}

In the case that $m = 2$, by definition,
\begin{equation}
\mathbf{W}_2(x) = \left(
\begin{array}{cc}
c_{11}x^{p(2)} & c_{12}x^{p(3)} \\
c_{21}x^{p(3)} & c_{22}x^{p(4)} \\
\end{array}
\right).
\end{equation}
Then
\begin{equation}
\vert\mathbf{W}_2(x)\vert = c_{11}c_{22}x^{p(2)+p(4)} - c_{12}c_{21}^2x^{2p(3)}.
\end{equation}
By (\ref{p1}),
\begin{equation}
p(2) + p(4) \leq 2p(3).
\end{equation}
Consequently, when $m=2$,
\begin{equation}
\begin{split}
\lim_{x\rightarrow0}\frac{\mathrm{log}\vert\mathbf{W}_{2}(x)\vert}{\mathrm{log}x} &= p(2) + p(4) \\
&= \sum_{i=1}^2\mathrm{min}\{a, 2i\}.
\end{split}
\end{equation}

Suppose when $m = k-1$, $k \in \mathbb{Z}^+ \cap [3,+\infty)$,
\begin{equation}
\lim_{x\rightarrow0}\frac{\mathrm{log}\vert\mathbf{W}_{k-1}(x)\vert}{\mathrm{log}x}  = \sum_{i=1}^{k-1}\mathrm{min}\{a, 2i\}.
\end{equation}
When $m=k$, $\mathbf{W}_k(x)$ can be written as
\begin{equation}
\left(
\begin{array}{cc}
\mathbf{W}_{k-1}(x) & \mathbf{b}_k(x) \\
\mathbf{b}_k^T(x) & c_{kk}x^{p(2k)} \\
\end{array}
\right)
\end{equation}
where the column vector
\begin{equation}
\mathbf{b}_k(x) = \left(
\begin{array}{c}
c_{1k}x^{p(k+1)} \\
\vdots \\
c_{k-1,k}x^{p(2k-1)} \\
\end{array}
\right).
\end{equation}
Hence, in terms of Schur determinant formula \cite{horn},
\begin{equation}
\begin{split}
\lim_{x\rightarrow 0}\frac{\mathrm{log}\vert\mathbf{W}_k(x)\vert}{\mathrm{log}x} &=
\lim_{x\rightarrow 0}\frac{\mathrm{log}\left[\vert \mathbf{W}_{k-1}(x)\vert\, \vert\mathbf{W}_{k-1}^*(x)\vert\right]}{\mathrm{log}x}
\\
&= \lim_{x\rightarrow 0}\frac{\mathrm{log}\vert\mathbf{W}_{k-1}(x)\vert}{
\mathrm{log}x}+\lim_{x\rightarrow 0}\frac{\mathrm{logdet}\mathbf{W}_{k-1}^*(x)}{\mathrm{log}x}
\label{lop1}
\end{split}
\end{equation}
where $\mathbf{W}_{k-1}^*(x)$ is the Schur complement of $\mathbf{W}_{k-1}(x)$,
\begin{equation}
\mathbf{W}_{k-1}^*(x) = c_{2k}x^{p(2k)} - \mathbf{b}^T_k(x)\mathbf{W}^{-1}_{k-1}(x) \mathbf{b}_k(x). \label{lop2}
\end{equation}
Since $\mathbf{W}_{k-1}(x)\mathbf{W}^{-1}_{k-1}(x) = \mathbf{I}$, $\mathbf{W}^{-1}_{k-1}(x)$ is of the form
\begin{equation}
\left(
\begin{array}{ccc}
c_{11}'x^{-p(2)} & \cdots & c_{1k}'x^{-p(k)} \\
\vdots & \ddots & \vdots \\
c_{k1}'x^{-p(k)} & \cdots & c_{k-1,k-1}'x^{-p(2k-2)} \\
\end{array}
\right).
\end{equation}
Consequently,
\begin{equation}
\begin{split}
\quad & \lim_{x\rightarrow 0}\frac{\mathrm{log}\left[\mathbf{b}^T_k(x)\mathbf{W}^{-1}_{k-1}(x) \mathbf{b}_k(x)\right]}{\mathrm{log}x} \\
\quad &= \mathrm{min}\{ p(2k-1)-p(k)+p(k+1),\quad p(2k-1)-p(k+1)+p(k+2), \quad \\
& \qquad \ldots, \quad p(2k-1)-p(2k-2)+p(2k-1)\} \\
&\stackrel{(a)}=p(2k-1)-p(2k-2)+p(2k-1)\\
&\stackrel{(b)}\geq p(2k)
\end{split}
\end{equation}
where both steps $(a)$ and $(b)$ follow the inequality (\ref{p1}).
Therefore, by (\ref{lop1}) and (\ref{lop2}),
\begin{equation}
\lim_{x\rightarrow 0}\frac{\mathrm{logdet}\mathbf{W}(x)}{\mathrm{log}x} = \sum_{i=1}^{k}\mathrm{min}\{a, 2i\},
\end{equation}
which concludes this proof.

\section{Proof of Lemma \ref{l2}}\label{appl2}

Each summand in $\vert \mathbf{W}(x)\vert$, which is a product of the entries $w_{1j_1},\ldots, w_{mj_m}$,  can be written as
\begin{equation}
x^{\sum_{k=1}^m(k+j_k)}\, \prod_{k=1}^m c_{k+j_k}
\end{equation}
where the numbers $\{j_1, j_2, ..., j_m\}$ is a permutation of $\{1, 2, ..., m \}$. Then, each summand has the same degree $m(m+1)$, which concludes the proof.

\section{Proof of Lemma \ref{l3}} \label{appl3}

By definition,
\begin{equation}
\mathbf{W} = \left(
\begin{array}{ccc}
\Gamma(a+1) & \cdots & \Gamma(a+m) \\
\vdots & \ddots & \vdots \\
\Gamma(a+m) & \cdots &  \Gamma(a+2m-1) \\
\end{array}
\right).
\end{equation}

For calculating the determinant of $\mathbf{W}$, we do Gaussian elimination by elementary row operations from bottom to top for obtaining the equivalent upper triangular $\mathbf{L}$\cite{hill}. Below-diagonal entries are eliminated from the first column to the last column.

Let $\mathbf{W}_l$ denote the matrix after the below-diagonal entries of the $l^{\mathrm{th}}$ column are eliminated. Then the $(i,j)^{\mathrm{th}}$ entry of $\mathbf{W}_{l}$ subject to $i \geq j > l$ is of the form
\begin{equation}
w_{l,i,j} = \theta_{l,i,j} \, \Gamma(a+i+j-1-l). \label{hlij}
\end{equation}
Hence, after below-diagonal entries of the $(l-1)^{\mathrm{th}}$ column are eliminated, for the entries subject to $i > l$ and $j=l$,
\begin{align}
w_{l-1, i-1, l} &= \theta_{l-1,i-1,l} \, \Gamma(a+i-1), \\
w_{l-1, i, l} &= \theta_{l-1,i,l} \, \Gamma(a+i).
\end{align}

Consequently, for eliminating the $\mathrm(i,l)^{\mathrm{th}}$ multiplied entry of $\mathbf{W}_{l-1}$ to obtain $\mathbf{W}_l$, the factor for the row operation in the Gaussian elimination on the $i^{\mathrm{th}}$ row
\begin{equation}
c_{l,i} = -\frac{\theta_{l-1, i, l}}{\theta_{l-1,i-1,l}}\, (a+i-1).
\end{equation}

That is, $w_{l,i,j}$ is obtained as follows:
\begin{equation}
\begin{split}
w_{l,i,j} &= w_{l-1, i, j} + c_{l,i}\, w_{l-1,i-1, j} \\
&= \left[\theta_{l-1,i,j}\,(a+i+j-l-1) - \theta_{l-1,i-1,j}\,\frac{\theta_{l-1, i,l}}{\theta_{l-1,i-1,l}}\,(a+i-1)\right]\\
& \quad \times \Gamma(a+i+j-l-1).
\end{split}
\end{equation}
Comparing the RHS of the above equation to (\ref{hlij}), we get
\begin{equation}
\theta_{l,i,j} = \theta_{l-1,i,j}\,(a+i+j-l-1) - \theta_{l-1,i-1,j}\,\frac{\theta_{l-1, i,l}}{\theta_{l-1,i-1,l}}\,(a+i-1).\label{tlij}
\end{equation}

Before doing any operation on $\mathbf{W}$, $\theta_{0,i,j} =1$. Then, by (\ref{tlij}), we obtain $\theta_{1,i,j} = j-1$ and $\theta_{2,i,j}= \Gamma(j)/\Gamma(j-2)$. Supposing
\begin{equation}
\theta_{l,i,j} = \frac{\Gamma(j)}{\Gamma(j-l)}.
\end{equation}
then by (\ref{tlij}) we have
\begin{equation}
\theta_{l+1,i,j} =  \frac{\Gamma(j)}{\Gamma(j-l-1)}.
\end{equation}
Therefore, our conjecture is right. Hence,
\begin{equation}
\theta_{i-1,i,i} = \Gamma(i).
\end{equation}
and the $i^{\mathrm{th}}$ diagonal entry of $\mathbf{L}$,
\begin{equation}
w_{i-1,i,i} = \Gamma(i)\Gamma(a+i).
\end{equation}
Consequently,
\begin{equation}
\vert\mathbf{W}_m\vert = \prod_{k=1}^m \Gamma(k)\Gamma(a+k),
\end{equation}
which concludes this proof.

\section{Proof of Lemma \ref{l4}}\label{appl4}

This proof is similar to Appendix \ref{appl3}.

By definition,
\begin{equation}
\mathbf{W} = \left(
\begin{array}{ccc}
\Gamma(a+1)\Gamma(b-1) & \cdots & \Gamma(a+m)\Gamma(b-m) \\
\vdots & \ddots & \vdots \\
\Gamma(a+m)\Gamma(b-m) & \cdots &  \Gamma(a+2m-1)\Gamma(b-2m+1) \\
\end{array}
\right).
\end{equation}

The $(i,j)^{\mathrm{th}}$ entry of $\mathbf{W}_{l}$ subject to $i \geq j > l$ is of the form
\begin{equation}
w_{l,i,j} = \theta_{l,i,j} \, \Gamma(a+i+j-1-l)\Gamma(b-i-j+1). \label{l4_hlij}
\end{equation}

Consequently, the multiplied factor
\begin{equation}
\begin{split}
c_{l,i} = -\frac{\theta_{l-1, i, l}\,(a+i-1)}{\theta_{l-1,i-1,l}\,(b-i-l+1)}.
\end{split}
\end{equation}
and
\begin{equation}
\begin{split}
w_{l,i,j} &= w_{l-1, i, j} + c_{l,i}\, w_{l-1,i-1, j} \\
&= \left[\theta_{l-1, i, j}\, (a+i+j-l-1) - \frac{\theta_{l-1,i-1,j}\, \theta_{l-1,i,l}\,(a+i-1)\,(b-i-j+1)}{\theta_{l-1,i-1,l}\, (b-i-l+1)}\right]\\
& \qquad \,\Gamma(a+i+j-l-1)\,\Gamma(b-i-j+1).
\end{split}
\end{equation}

Comparing the RHS of the above expression to (\ref{l4_hlij}), we get
\begin{equation}
\theta_{l,i,j} = \theta_{l-1,i,j}\,(a+i+j-l-1) - \theta_{l-1,i-1,j}\,\frac{\theta_{l-1,i,l}\,(a+i-1)\,(b-i-j+1)} {\theta_{l-1,i-1,l}(b-i-l+1)}
\label{l4_tlij}
\end{equation}

Before doing any operation on $\mathbf{W}$, $\theta_{0,i,j} =1$. Then, by (\ref{l4_tlij}), we obtain
\begin{align}
\theta_{1,i,j} &= \frac{(j-1)(a+b-1)}{(b-i)}, \\
\theta_{2,i,j} &= \frac{(j-1)(j-2)(a+b-1)(a+b-2))}{(b-i)(b-i-1)}.
\end{align}
Supposing
\begin{equation}
\theta_{l,i,j} = \prod_{k=1}^l\frac{(j-k)(a+b-k)}{(b-i-l+k)}.
\end{equation}
then by (\ref{l4_tlij}) we have
\begin{equation}
\theta_{l+1,i,j} = \prod_{k=1}^{l+1}\frac{(j-k)(a+b-k)}{(b-i-l+k)}.
\end{equation}
Therefore, our conjecture is right. Hence, for $i\geq 2$, the $i^{\mathrm{th}}$ diagonal entry of the equivalent upper triangular $\mathbf{L}$,
\begin{equation}
w_{i-1,i,i} = \Gamma(a+b)\,\Gamma(i)\,\Gamma(a+i)\frac{\Gamma(b-2i+2)\Gamma(b-2i+1)}
{\Gamma(a+b-i+1)\Gamma(b-i+1)}.
\end{equation}
Consequently,
\begin{equation}
\begin{split}
\vert\mathbf{W}\vert &= \Gamma(a+1)\Gamma(b-1)\Gamma^{m-1}(a+b)\\
& \quad \, \prod_{k=2}^m  \Gamma(k)\Gamma(a+k)\frac{\Gamma(b-2k+2)\Gamma(b-2k+1)}{\Gamma(a+b-k+1)\Gamma(b-k+1)},
\end{split}
\end{equation}
which concludes this proof.

\section{Proof of Lemma \ref{l8}}\label{appl8}
The derivation of Lemma \ref{l8} is analogous to Appendix \ref{appl3}. However, for deriving Lemma \ref{l8}, we use Gaussian elimination by column operations from the right to the left, instead of row operations from the bottom to the top in Appendix \ref{appl3}. After the Gaussian elimination, the left upper-diagonal triangle-matrix becomes a zero triangle-matrix. Consequently, the determinant of $\mathbf{W}$ is
\begin{equation}
\vert\mathbf{W}\vert = (-1)^{\frac{m(m-1)}{2}}\prod_{k=1}^m \Gamma(k)\Gamma(a+k-m).
\end{equation}

\section{Proof of Lemma \ref{l6}}\label{appl6}
$f(n)$ can be written as
\begin{equation}
f(n) = \frac{\Gamma(n-a)}{\Gamma(n)}\cdots\frac{\Gamma(n-m+1-a)}{\Gamma(n-m+1)}.
\end{equation}
We thus have
\begin{equation}
f(n+1) - f(n) = \left(\frac{n-a}{n}\cdots\frac{n-m+1-a}{n-m+1}-1\right)\, f(n).
\end{equation}
It is seen that $\frac{n-a}{n}\cdots\frac{n-m+1-a}{n-m+1} < 1$ and $f(n) > 0$. Hence, $f(n+1)-f(n) < 0$, \emph{i.e.}, $f(n)$ is monotonically decreasing.

For $g(n)$,
\begin{equation}
\begin{split}
g(n+1)-g(n) &= \left((n+1)^{am}\;\frac{n-a}{n}\cdots\frac{n-m+1-a}{n-m+1} - n^{am} \right)\, f(n) \\
&\leq \left[(n+1)^{am}\;\left(\frac{n-a}{n}\right)^m - n^{am} \right]\, f(n)
\end{split}
\end{equation}

If
\begin{equation}
(n+1)^a\cdot \frac{n-a}{n} < n^a,
\end{equation}
then we have $g(n+1)-g(n) < 0$.

Define a function $h(x)$,
\begin{equation}
\begin{split}
h(x) &= (x-a)(x+1)^a-x^{a+1}\\
&= (x+1)^{a+1} - x^{a+1} - (a+1)(x+1)^a,\qquad x>a
\end{split}
\end{equation}
In terms of mean value theory \cite{rudin}, for $\phi(x) = x^{a+1}$, there exists $\xi$ which lets
\begin{equation}
\phi'(\xi) = (x+1)^{a+1} - x^{a+1}, \qquad x < \xi < x+1
\end{equation}
where $\phi'(\xi)$ is the first derivative.

As
\begin{equation}
\phi^{''}(x) = a(a+1)x^{a-1} > 0,
\end{equation}
$\phi^{'}(x)$ is monotonically increasing and thus
\begin{equation}
\phi^{'}(\xi) < \phi^{'}(x+1).
\end{equation}
So,  $h(x) < 0$.

Then, we have
\begin{equation}
\frac{x-a}{x} < \left(\frac{x}{x+1}\right)^a.
\end{equation}
When $x = n$,
\begin{equation}
(n+1)^a\,\frac{n-a}{n} < n^a
\end{equation}
Consequently, $g(n+1)-g(n) < 0$, that is, $g(n)$ is monotonically decreasing.

\section{Proof of Lemma \ref{l7}}\label{appl7}

In terms of Euler's reflection formula
\begin{equation}
\Gamma(1-x)\Gamma(x) = \frac{\pi}{\sin(\pi x)},
\end{equation}
\begin{align}
\Gamma(a+n+1)\Gamma(-a-n) &= \frac{\pi}{\sin\left(\pi (a+n+1)\right)}, \\
\Gamma(a+1)\Gamma(-a) &= \frac{\pi}{\sin\left(\pi (a+1)\right)}.
\end{align}
Straightforwardly,
\begin{equation}
\frac{\Gamma(a+n+1)}{\Gamma(a+1)} = (-1)^n \frac{\Gamma(-a)}{\Gamma(-a-n)},
\end{equation}
i.e.,
\begin{equation}
(a+1)_n = (-1)^n(-a-n)_n.
\end{equation}

\section{Proof of Theorem \ref{t3}} \label{appunc}

From the proof of Theorem \ref{t2}, we see that
\begin{equation}
\mu^*_{\mathrm{unc}}(\eta) = \frac{P_s\vert\mathbf{E}(\eta)\vert }{\prod_{k=1}^{N_{\mathrm{min}}}\Gamma(N_\mathrm{max} -k+1)\Gamma(N_{\mathrm{min}}-k+1)} \label{muunc1}
\end{equation}
where $\mathbf{E}(\eta)$ is an $N_{\min}\times N_{\min}$ matrix of $e_{ij}(\eta)$'s.
\\

\begin{compactenum}
\item When $2/\eta \in (0, \vert N_t-N_r\vert+1)$, given by (\ref{uij}) and Table \ref{tb_psisx}, we have
\begin{equation}
e_{ij}(\eta) = {N_t}^{\frac{2}{\eta}}\Gamma(d_{ij}-\frac{2}{\eta}).
\end{equation}
By Lemma \ref{l3},
\begin{equation}
\vert\mathbf{E}(\eta)\vert = N_t^{\Delta_{\mathrm{unc}}^*}\kappa_h\left(\frac{2}{\eta}, N_{\min}, N_{\min}, N_{\max}\right). \label{muunc2}
\end{equation}
In this case, $\Delta^*_{\mathrm{unc}}(\eta) = 2N_{\min}/\eta$. Substituting (\ref{muunc2}) into (\ref{muunc1}), we obtain the optimum distortion factor in this case in the closed form
\begin{equation}
\mu_{\mathrm{unc}}^*(\eta) = P_s{N_t}^{\Delta^*_{\mathrm{unc}}}\frac{\kappa_h(\frac{2}{\eta}, N_{\min}, N_{\min}, N_{\max} )}{\prod_{k=1}^{N_{\min}}\Gamma(N_{\max}-k+1)\Gamma(N_{\min}-k+1)}.
\end{equation}
In the light of Lemma \ref{l6}, it monotonically decreases with $N_{\max}$.
\\
\item When $2/\eta \in (N_t+N_r-1, \infty)$, in terms of (\ref{uij}) and Table \ref{tb_psisx}, we have
\begin{equation}
e_{ij}(\eta) = N_t^{d_{ij}}\Gamma(d_{ij})\frac{\Gamma\;\left(\frac{2}{\eta}-d_{ij}\right)} {\Gamma\left(\frac{2}{\eta}\right)}.\label{eij2}
\end{equation}
In terms of Lemma \ref{l2} and Lemma \ref{l4}, the determinant of $\mathbf{E}(\eta)$ is
\begin{equation}
\vert\mathbf{E}(\eta)\vert = N_t^{\Delta_{\mathrm{unc}}^*}\kappa_l\left(\frac{2}{\eta}, N_{\min}, N_{\min}, N_{\max} \right). \label{muunc3}
\end{equation}
In this case, $\Delta^*_{\mathrm{unc}}(\eta) = N_tN_r$. Substituting (\ref{muunc3}) into (\ref{muunc1}), we obtain the optimum distortion factor in this case in the form
\begin{equation}
\mu_{\mathrm{unc}}^* = P_s{N_t}^{\Delta^*_{\mathrm{unc}}}\frac{\kappa_l(\frac{2}{\eta}, N_{\min}, N_{\min}, N_{\max} )}{\prod_{k=1}^{N_{\min}}\Gamma(N_{\max}-k+1)\Gamma(N_{\min}-k+1)}.
\end{equation}
\\
\item When $2/\eta \in [\vert N_t-N_r\vert+1, N_t+N_r-1]$, the analysis is relatively complex.  Define a partition number
\begin{equation}
l = \left\lfloor \frac{\frac{2}{\eta}+1-\vert N_t-N_r\vert}{2}\right\rfloor
\end{equation}
and partition the Hankel matrix $\mathbf{E}(\eta)$ in (\ref{Dlb}) as
\begin{equation}
\mathbf{E}(\eta) = \left(
\begin{array}{cc}
\mathbf{A} & \mathbf{B} \\
\mathbf{B}^{\mathrm{T}} & \mathbf{C} \\
\end{array}
\right)
\end{equation}
where $\mathbf{A}$ is the $l\times l$ submatrix and $\mathbf{C}$ is the $(N_{\min}-l)\times (N_{\min}-l)$ submatrix.

At high SNR, in terms of Table \ref{tb_psisx},  if $2l \neq \frac{2}{\eta}+1-\vert N_t - N_r\vert$, entries of $\mathbf{A}$ and $\mathbf{C}$ approximate
\begin{align}
\tilde{a}_{ij} &= N_t^{d_{ij}}\Gamma(d_{ij})\frac{\Gamma(\frac{2}{\eta}-d_{ij})} {\Gamma(\frac{2}{\eta})}\;\rho^{-d_{ij}}, \\
\tilde{c}_{ij} & = N_t^{\frac{2}{\eta}}\Gamma(d_{ij}-\frac{2}{\eta})\rho^{-\frac{2}{\eta}} \label{c2eq:cij};
\end{align}
if $2l = \frac{2}{\eta}+1-\vert N_t - N_r\vert$, the form of $\tilde{c}_{ij}$ is the same as (\ref{c2eq:cij}) whereas the form of $\tilde{a}_{ij}$ becomes
\begin{equation}
\tilde{a}_{ij} = \begin{cases} N_t^{d_{ij}}\Gamma(d_{ij})\frac{\Gamma(\frac{2}{\eta}-d_{ij})} {\Gamma(\frac{2}{\eta})}\;\rho^{-d_{ij}}, \quad (i,j) \neq (l,l); \\
N_t^{\frac{2}{\eta}}\log\rho\;\rho^{-\frac{2}{\eta}}, \quad (i,j) = (l,l).
\end{cases}
\end{equation}

In terms of Schur determinant formula \cite{horn},
\begin{equation}
\vert \mathbf{E}(\eta)\vert = \vert\mathbf{A}\vert \vert\mathbf{C}-\mathbf{A}^*\vert
\end{equation}
where $\mathbf{A}^* = \mathbf{B}^T\mathbf{A}^{-1}\mathbf{B}$. By the method analogous to the derivation in Appendix \ref{appl1}, we know that for high SNR
\begin{equation}
\mathbf{C} - \mathbf{A}^* \sim \widetilde{\mathbf{C}}
\end{equation}
where $\widetilde{\mathbf{C}}$ is composed of $\tilde{c}_{ij}$'s. Consequently,
\begin{equation}
\vert\mathbf{E}(\eta)\vert \sim \vert\widetilde{\mathbf{A}}\vert\vert\widetilde{\mathbf{C}}\vert.
\end{equation}
Given the preceding derivation for high and low SCBR regimes, we have
\begin{align}
\vert\widetilde{\mathbf{A}}\vert & = \begin{cases}
N_t^{l(l+N_{\max}-N_{\min})}\kappa_l(\frac{2}{\eta}, l, N_{\min}, N_{\max})\rho^{-l(l+N_{\max}-N_{\min})}, \\
\quad \text{if}\; 2l \neq \frac{2}{\eta}+1-\vert N_t - N_r\vert; \\
N_t^{l(l+N_{\max}-N_{\min})}\kappa_l(\frac{2}{\eta}, l-1, N_{\min}, N_{\max})\log\rho\;\rho^{-l(l+N_{\max}-N_{\min})}, \\
\quad \text{if}\; 2l = \frac{2}{\eta}+1-\vert N_t - N_r\vert,
\end{cases}
\\
\vert\widetilde{\mathbf{C}}\vert& = N_t^{\frac{2(N_{\min}-l)}{\eta}}\kappa_h(\frac{2}{\eta}-2l, N_{\min}-l, N_{\min}, N_{\max})\rho^{-\frac{2(N_{\min}-l)}{\eta}}.
\end{align}
Therefore, in this case,
\begin{equation}
\mu^*_{\mathrm{unc}}(\eta) = \begin{cases}
P_sN_t^{\Delta^*_{\mathrm{unc}}}\frac{\kappa_l(\frac{2}{\eta}, l, N_{\min}, N_{\max})\kappa_h(\frac{2}{\eta}-2l, N_{\min}-l, N_{\min}, N_{\max})}{\prod_{k=1}^{N_{\min}}\Gamma(N_{\max}-l+1)\Gamma(N_{\min}-k+1)}, \\
\quad 2l \neq \frac{2}{\eta}+1-\vert N_t - N_r\vert; \\
P_sN_t^{\Delta^*_{\mathrm{unc}}}\log\rho\frac{\kappa_l(\frac{2}{\eta}, l-1, N_{\min}, N_{\max})\kappa_h(\frac{2}{\eta}-2l, N_{\min}-l, N_{\min}, N_{\max})}{\prod_{k=1}^{N_{\min}}\Gamma(N_{\max}-l+1)\Gamma(N_{\min}-k+1)}, \\
\quad 2l = \frac{2}{\eta}+1-\vert N_t - N_r\vert
\end{cases}
\end{equation}
where the optimum distortion exponent is
\begin{equation}
\Delta^*_{\mathrm{unc}}(\eta) = l(l+\vert N_t-N_r\vert) +\frac{2(N_{\min}-l)}{\eta}.
\end{equation}

This concludes the proof of this theorem.

\end{compactenum}

\section{Proof of Theorem \ref{t5}} \label{appt5}
Let $\widetilde{\mathbf{G}}$ denote the asymptotic form of $\mathbf{G}$ for high SNR. Since $g_{ij}$ is a polynomial of $\rho^{-1}$ given by (\ref{gij}) and the preliminaries in Section \ref{sec_pre}, in terms of Table \ref{tb_psisx}, $\vert\widetilde{\mathbf{G}}\vert$ can be written as $\sum_{m=1}^M\vert\widetilde{\mathbf{G}}_m\vert$ where
\begin{equation}
\vert\widetilde{\mathbf{G}}_m\vert = u_m \rho^{-\Delta_{\mathrm{cor}}^*}, \label{c3_defu}
\end{equation}
i.e., they have the same degree over $\rho^{-1}$. Each entry of $\widetilde{\mathbf{G}}_m$ is a monomial of $\rho^{-1}$ denoted by $\widetilde{g}_{m, ij}$. In terms of Table \ref{tb_psisx} and the preliminaries in Section \ref{sec_pre}, we learn that $\widetilde{g}_{m,ij}$'s form is one of $\sigma_i^{-r_{m,j}}a(j,r_{m,j})\rho^{-(d_j+r_{m,j})}$ (Form 1) and $\sigma_i^{d_j-\frac{2}{\eta}}c_{j}\log^{\epsilon}\rho\;\rho^{-\frac{2}{\eta}}$ (Form 2), where $r_{m,j}$ is a non-negative integer, $\epsilon = 0,1$, and
\begin{align}
a(j, r_{m,j}) &= N_t^{d_j+r_{m,j}}\frac{\Gamma(\frac{2}{\eta}-d_j)\Gamma(d_j+r_{m,j})} {\Gamma(\frac{2}{\eta})\Gamma(r_{m,j}+1)(d_j+1-\frac{2}{\eta})_{r_{m,j}}}\label{ajr1}\\
c_j &= N_t^{\frac{2}{\eta}}\Gamma(d_j-\frac{2}{\eta})\label{cj1}.
\end{align}

If the entries of first $l$ columns of $\widetilde{\mathbf{G}}_m$ are of Form 1 and other entries are of Form 2, $\widetilde{\mathbf{G}}_m$ can be partitioned as
\begin{equation}
\widetilde{\mathbf{G}}_m = \left(
\begin{array}{cc}
\widetilde{\mathbf{G}}_{m,1} & \widetilde{\mathbf{G}}_{m,2} \\
\end{array}
\right)
\end{equation}
where $\widetilde{\mathbf{G}}_{m,1}$ is of size $N_{\min}\times l$ and $\widetilde{\mathbf{G}}_{m,2}$ is of size $N_{\min}\times (N_{\min}-l)$. Since $\widetilde{\mathbf{G}}_m$ is a full-rank matrix, $\widetilde{\mathbf{G}}_{m,1}$ and $\widetilde{\mathbf{G}}_{m,2}$ ought to be full rank as well. Apparently, $\widetilde{\mathbf{G}}_{m,2}$ is a full-rank matrix; whereas, for $\widetilde{\mathbf{G}}_{m,1}$, if there exist $r_{m,j_1} = r_{m,j_2}$ for $j_1 \neq j_2$, $\widetilde{\mathbf{G}}_{m,1}$ would not be full rank, because in that case, its submatrix constructed by the two columns with individual indices $j_1$ and $j_2$ would be rank-one. Thus, each $r_{m,j}$ must be distict.

Now let us figure out $l$. Define a distortion exponent function as \begin{equation}
\gamma(n) = \begin{cases}
\sum_{k=1}^n d_k + {\sum_{k=0}^{n-1}k+\frac{2(N_{\min}-n)}{\eta}}, \quad & n\in \mathbb{Z} \cap (0, N_{\min}]; \\
\frac{2N_{\min}}{\eta}, \quad & n = 0.
\end{cases}
\end{equation}
Apparently, $\gamma(n)$ is on the curve of the two-order function $f(x)$,
\begin{equation}
f(x) = x^2 + \left(\vert N_t-N_r\vert-\frac{2}{\eta}\right)x + \frac{2N_{\min}}{\eta}
\end{equation}
which is a symmetric convex function and whose minimum value is given by $x = \frac{\frac{2}{\eta}-\vert N_t-N_r\vert}{2}$.

Since $n=l$ gives the minimum $\gamma(n)$, when $2/\eta \in (0, \vert N_t-N_r\vert+1)$, $l = 0$, $\Delta_{\mathrm{cor}}(\eta) = \gamma(0) = 2N_{\min}/\eta$; when $2/\eta \in (N_t+N_r-1, +\infty)$, $l = N_{\min}$, $\Delta_{\mathrm{cor}}(\eta) = \gamma(N_{\min}) = N_tN_r$.

When $\eta \in [\vert N_t-N_r\vert+1, N_t+N_r-1]$, we should have
\begin{equation}
\gamma(l) \leq \gamma(l-1)
\end{equation}
and
\begin{equation}
\gamma(l) \leq \gamma(l+1),
\end{equation}
which gives
\begin{equation}
\frac{2}{\eta}-1-\vert N_t-N_r\vert \leq 2l \leq \frac{2}{\eta}+1-\vert N_t-N_r\vert.
\end{equation}
Hence, for $\eta \in [\vert N_t-N_r\vert+1, N_t+N_r-1]$,
\begin{equation}
l = \left\lfloor\frac{\frac{2}{\eta}+1-\vert N_t-N_r\vert}{2}\right\rfloor \quad \text{or} \quad \left\lceil \frac{\frac{2}{\eta}-1-\vert N_t-N_r\vert}{2}\right\rceil
\end{equation}
and
\begin{equation}
\begin{split}
\Delta^*_{\mathrm{cor}}(\eta) &= \gamma(l) \\
&= l(l+\vert N_r-N_t\vert)+\frac{2(N_{\min}-l)}{\eta} \\
& = \sum_{k=1}^{N_{\min}}{\min\left\{\frac{2}{\eta}, 2k-1+\vert N_t-N_r \vert\right\}}.
\end{split}
\end{equation}
Note that $\gamma\left(\left\lfloor\frac{\frac{2}{\eta}+1-\vert N_t-N_r\vert}{2}\right\rfloor\right) = \gamma\left(\left\lceil \frac{\frac{2}{\eta}-1-\vert N_t-N_r\vert}{2}\right\rceil\right)$.

This concludes the proof of this theorem.

\section{Proof of Theorem \ref{t6}} \label{appt6}
From the proofs of Theorem \ref{t4} and Theorem \ref{t5}, we have
\begin{equation}
\mu^*_{\mathrm{cor}} = \frac{P_s\vert\bm{\Sigma}\vert^{-N_{\max}}\sum_{m=1}^{M}u_m} {\prod_{k=1}^{N_{\mathrm{min}}}\Gamma(N_{\max} -k+1)\vert\mathbf{V}_2(\bm{\sigma})\vert}
\end{equation}
where $u_m$ is defined in (\ref{c3_defu}).

\begin{compactenum}
\item Consider the case of $2/\eta \in (0, \vert N_t-N_r\vert +1)$. We have $M=1$ and
\begin{equation}
\widetilde{g}_{1,ij} = \sigma_i^{d_j-\frac{2}{\eta}}c_j\rho^{-\frac{2}{\eta}}, \quad i = 1, \ldots N_{\min},\; j = 1, \ldots N_{\min}
\end{equation}
where $d_j$ is defined in Theorem \ref{t4} and $u_j$ is defined in (\ref{cj1}).
Thereby,
\begin{equation}
u_1 = N_t^{\frac{2N_{\min}}{\eta}}\vert\mathbf{V}_1(\bm{\sigma})\vert \prod_{j=1}^{N_{\min}}\Gamma(d_j-\frac{2}{\eta}) \prod_{i=1}^{N_{\min}}\sigma_i^{\vert N_t-N_r\vert+1-\frac{2}{\eta}}.
\end{equation}
So, in this case,
\begin{equation}
\begin{split}
\mu_{\mathrm{cor}}^*(\eta) &= \frac{\vert\bm{\Sigma}\vert^{-N_{\max}}\vert \mathbf{V}_1(\bm{\sigma})\vert \prod_{i=1}^{N_{\min}}\sigma_i^{\vert N_t-N_r\vert+1-\frac{2}{\eta}}}{\vert\mathbf{V}_2(\bm{\sigma})\vert}\\ & \quad \times \frac{P_sN_t^{\frac{2N_{\min}}{\eta}}\prod_{j=1}^{N_{\min}}\Gamma(d_j-\frac{2}{\eta})} {\prod_{k=1}^{N_{\mathrm{min}}}\Gamma(N_{\max} -k+1)} \\
&= \prod_{k=1}^{N_{\min}}\sigma_k^{-\frac{2}{\eta}}\, \mu^*_{\mathrm{unc}}(\eta).
\end{split}
\end{equation}
Note that $\mathbf{V}_1(\bm{\sigma})$ and $\mathbf{V}_2(\bm{\sigma})$ are Vandermonde matrices defined by (\ref{c3_defv1}) and (\ref{c3_defv2}) respectively in the proof of Theorem \ref{t4}.
\item Consider the case of $2/\eta \in (N_t + N_r -1, +\infty)$. We have $M = N_{\min}!$ and
\begin{equation}
\begin{split}
\widetilde{g}_{m,ij}= \sigma_i^{-r_{m,j}}a(j, r_{m,j})\rho^{-d_j-r_{m,j}}, &\quad m = 1, \dots, M, \; i = 1,\ldots,N_{\min}, \\
& \quad \; j = 1,\ldots,N_{\min}
\end{split}
\end{equation}
where
\begin{equation}
\begin{split}
a(j, r_{m,j}) &= N_t^{d_j+r_{m,j}}\frac{\Gamma(d_j)\Gamma(\frac{2}{\eta}-d_j)(d_j)_{r_{m,j}}} {\Gamma(\frac{2}{\eta})\Gamma(r_{m,j}+1)(d_j+1-\frac{2}{\eta})_{r_{m,j}}}\\
&= N_t^{d_j+r_{m,j}}\frac{\Gamma(\frac{2}{\eta}-d_j)\Gamma(d_j+r_{m,j})} {\Gamma(\frac{2}{\eta})\Gamma(r_{m,j}+1)(d_j+1-\frac{2}{\eta})_{r_{m,j}}}
\end{split}\label{aj1}
\end{equation}
By Lemma \ref{l8},
\begin{equation}
\left(d_j+1-\frac{2}{\eta}\right)_{r_{m,j}} = (-1)^{r_{m,j}}\left(\frac{2}{\eta}-d_j-r_{m,j}\right)_{r_{m,j}}.\label{aj2}
\end{equation}
Substitute (\ref{aj2}) to (\ref{aj1}), we have
\begin{equation}
a(j, r_{m,j}) = (-1)^{r_{m,j}} N_t^{d_j+r_{m,j}}\frac{\Gamma(d_j+r_{m,j})\Gamma(\frac{2}{\eta}-d_j-r_{m,j})} {\Gamma(\frac{2}{\eta})\Gamma(r_{m,j}+1)}.\label{ajr}
\end{equation}

Hence,
\begin{equation}
\begin{split}
u_m &= (-1)^{\sum_jr_{m,j}}\mathrm{sgn}(\mathbf{r}_m)\vert\mathbf{V}_2(\bm{\sigma})\vert\; \prod_{j=1}^{N_{\min}}a(j, r_{m,j})\\
&= \mathrm{sgn}(\mathbf{r}_m)\vert\mathbf{V}_2(\bm{\sigma})\vert\;  \prod_{j=1}^{N_{\min}}N_t^{d_j+r_{m,j}} \frac{\Gamma(d_j+r_{m,j})\Gamma(\frac{2}{\eta}-d_j-r_{m,j})} {\Gamma\left(\frac{2}{\eta}\right)\Gamma(r_{m,j}+1)}
\end{split}
\end{equation}
Note that $\mathbf{r}_m$ is a permutation of $\{0, 1, \ldots, N_{\min}-1\}$ and $\mathrm{sgn}(\mathbf{r}_m)$ denotes the signature of the permutation $\mathbf{r}_m$: $+1$ if $\mathbf{r}_m$ is an even permutation and $-1$ if $\mathbf{r}_m$ is an odd permutation.

Consequently, in the light of Leibniz formula \cite{horn},
\begin{equation}
\sum_{m=1}^M u_m = \frac{\vert\mathbf{V}_2(\sigma)\vert}{\prod_{k=1}^{N_{\min}}\Gamma(k)}\; \vert\mathbf{Q}\vert
\end{equation}
where each entry of $\mathbf{Q}$ is
\begin{equation}
q_{ij} = N_t^{d_{ij}}\Gamma(d_{ij})\frac{\Gamma(\frac{2}{\eta}-d_{ij})}{\Gamma(\frac{2}{\eta})}. \label{qij2}
\end{equation}
Note that $d_{ij}$ is defined in the description of Theorem \ref{t1}. Comparing (\ref{qij2}) to (\ref{eij2}), we find that $q_{ij}$ and $e_{ij}$ are identical. Therefore,
\begin{equation}
\mu_{\mathrm{cor}}^*(\eta) = \prod_{k=1}^{N_{\min}}\sigma_k^{-N_{\max}}\, \mu^*_{\mathrm{unc}}(\eta).
\end{equation}
\item  Consider the case of $2/\eta \in [\vert N_t - N_r\vert -1, N_t+N_r +1]$. In terms of the proof of Theorem \ref{t5} and the preliminaries in Section \ref{sec_pre}, when {$\mod\{2/\eta+1-\vert N_t-N_r\vert,2\} \neq 0$}, $M = l!$,
\begin{equation}
\tilde{g}_{m,ij}=\begin{cases}
\sigma_i^{-r_{m,j}}a(j, r_{m,j})\rho^{-d_j-r_{m,j}},\quad &j \leq l; \\
\sigma_i^{d_j-\frac{2}{\eta}}c_{j}\rho^{-\frac{2}{\eta}}, \quad &j\geq l+1;
\end{cases}
\end{equation}
when $\mod\{2/\eta+1-\vert N_t-N_r\vert,2\} = 0$, $M = (l-1)!$,
\begin{equation}
\tilde{g}_{m,ij}=\begin{cases}
\sigma_i^{-r_{m,j}}a(j, r_{m,j})\rho^{-d_j-r_{m,j}},\quad &j \leq l-1; \\
\sigma_i^{-l+1}(-1)^{l-1}\frac{N_t^{\frac{2}{\eta}}}{\Gamma(l)}\log\rho \,\rho^{-\frac{2}{\eta}} , \quad &j = l; \\
\sigma_i^{d_j-\frac{2}{\eta}}c_{j}\rho^{-\frac{2}{\eta}}, \quad &j\geq l+1.
\end{cases}
\end{equation}
Note that $a(j,r_{m,j})$ and $c_{j}$ are given by (\ref{ajr1}) and (\ref{cj1}) respectively; when $\mod\{2/\eta+1-\vert N_t-N_r\vert,2\} \neq 0$, $\mathbf{r}_m$ is a permutation of $\{0, 1, \ldots, l-1 \}$; when $\mod\{2/\eta+1-\vert N_t-N_r\vert,2\} = 0$, $\mathbf{r}_m$ is a permutation of $\{0, 1, \ldots, l-2 \}$. Thus,
\begin{equation}
u_m = \begin{cases} \mathrm{sgn}(\mathbf{r}_m)\vert\mathbf{V}_3(\mathbf{\bm{\sigma}})\vert \prod_{j=1}^l a(j, r_{m,j}) \prod_{j=l+1}^{N_{\min}}N_t^{\frac{2}{\eta}}\Gamma(d_j-\frac{2}{\eta}),\\
\quad \mod\{2/\eta+1-\vert N_t-N_r\vert,2\} \neq 0; \\
\mathrm{sgn}(\mathbf{r}_m)\vert\mathbf{V}_3(\mathbf{\bm{\sigma}})\vert (-1)^{l-1}N_t^{\frac{2(N_{\min}-l+1)}{\eta}}\log\rho\\
\quad \times\prod_{j=1}^{l-1} a(j, r_{m,j}) \prod_{j=l+1}^{N_{\min}}\Gamma(d_j-\frac{2}{\eta}),\\
\quad \mod\{2/\eta+1-\vert N_t-N_r\vert,2\} = 0.
\end{cases}
\end{equation}
where each entry of $\mathbf{V}_3(\bm{\sigma})$,
\begin{equation}
v_{3,ij} = \sigma_i^{-\min\{j-1, \frac{2}{\eta}-d_j\}}.
\end{equation}
Comparing to the proof of Theorem \ref{t3} for the same case of $\eta$, we have
\begin{equation}
\begin{split}
\mu_{\mathrm{cor}}^*(\eta) &=
\frac{(-1)^{\frac{l(l-1)}{2}}\vert\mathbf{V}_3(\bm{\sigma})\vert} {\prod_{k=1}^{N_{\min}}\sigma_k^{\vert N_t-N_r\vert+1}\prod_{1\leq m < n \leq N_{\min}}(\sigma_n-\sigma_m)} \\
&\quad \times\prod_{k=1}^{N_{\min}-l}\frac{(k)_l}{(\vert N_t-N_r\vert-\frac{2}{\eta}+l+k)_l}\,\mu^*_{\mathrm{unc}}(\eta).
\end{split}
\end{equation}
\end{compactenum}
This concludes the proof.

\section{Proof of Theorem \ref{t7}} \label{appt7}
When $2/\eta \in (0, \vert N_t-N_r\vert +1)$ or $2/\eta \in (N_t + N_r -1, +\infty)$, in terms of Theorem \ref{t6}, straightforwardly,  $\lim_{\bm{\Sigma}\rightarrow \mathbf{I}}\mu^*_{\mathrm{cor}}(\eta) = \mu^*_{\mathrm{unc}}(\eta)$  .

Consider the case of $2/\eta \in [\vert N_t - N_r\vert -1, N_t+N_r +1]$. By Taylor expansion and Lemma \ref{l8} , the entries of $\mathbf{V}_3(\bm{\sigma})$
\begin{equation}
\begin{split}
v_{3, ij} &= \sum_{n=0}^{\infty}\frac{(-p_j-n+1)_n}{n!}(\sigma_i-1)^n \\
&=  \sum_{n=0}^{\infty}\frac{(-1)^n(p_j)_n}{n!}(\sigma_i-1)^n
\end{split}
\end{equation}
where $p_j = \min\{j-1, \frac{2}{\eta}-d_j\}$.

Thereby, when $\bm{\sigma}$ approaches a vector of ones,
\begin{equation}
\vert\mathbf{V}_3(\bm{\sigma})\vert = \sum_{m=1}^{(N_{\min}-1)!}\vert \mathbf{V}_{3,m}(\bm{\sigma})\vert
\end{equation}
where the entries of $ \mathbf{V}_{3,m}(\bm{\sigma})$
\begin{equation}
v_{3,m,ij} = \begin{cases}
1, \quad & j = 1; \\
\frac{(-1)^{s_{m,j}}(p_j)_{s_{m,j}}}{s_{m,j}!}(\sigma_i-1)^{s_{m,j}}, \quad &j \geq 1.
\end{cases}
\end{equation}
Note that $\mathbf{s}_m = \{s_{m,2},\ldots, s_{m,{N_{\min}}}\}$ is a permutation of $\{1, 2, \ldots, N_{\min}-1\}$.

The determinant of $\mathbf{V}_{3,m}(\bm{\sigma})$
\begin{equation}
\vert \mathbf{V}_{3,m}(\bm{\sigma})\vert = (-1)^{n_1}\vert\mathbf{V}_1(\bm{\sigma}-\bm{1})\vert \mathrm{sgn}(\mathbf{s}_m)\prod_{k=2}^{N_{\min}}\frac{1}{\Gamma(p_k)\Gamma(k)} \prod_{j=2}^{N_{\min}}\Gamma(s_{m,j}+p_j)
\end{equation}
where $n_1 = \frac{N_{\min}(N_{\min}-1)}{2}$.
In the light of Leibniz formula  \cite{horn} and
\begin{equation}
\vert\mathbf{V}_1(\bm{\sigma}-\mathbf{a})\vert = \vert\mathbf{V}_1(\bm{\sigma})\vert, \quad \mathbf{a}  = \{a,\ldots,a\},
\end{equation}
$\vert \mathbf{V}_3(\bm{\sigma})\vert$ can be written in the form
\begin{equation}
\vert \mathbf{V}_3(\bm{\sigma})\vert = (-1)^{\frac{N_{\min}(N_{\min}-1)}{2}}\vert \mathbf{V}_1(\bm{\sigma})\vert\vert \mathbf{W}\vert\prod_{k=2}^{N_{\min}}\frac{1}{\Gamma(p_k)\Gamma(k)}
\end{equation}
where $\mathbf{W}$ is an $(N_{\min}-1)\times (N_{\min}-1)$ matrix with entries
\begin{equation}
\begin{split}
w_{ij} &= \Gamma(i+p_{j+1}) \\
&= \begin{cases}
\Gamma(i+j), \quad &j \leq l-1\\
\Gamma\left(\frac{2}{\eta}-\vert N_t-N_r\vert-1+i-j\right), \quad &j \geq l.
\end{cases}
\end{split}
\end{equation}

By partial Gaussian elimination, $\mathbf{W}$ can be transformed to $\mathbf{W}^{'}$ with a $(N_{\min}-l)\times (l-1)$ left-lower submatrix of zeros. Partition $\mathbf{W}^{'}$ as
\begin{equation}
\mathbf{W}^{'} = \left(
\begin{array}{cc}
\mathbf{W}_1^{'} & \mathbf{W}_2^{'} \\
\mathbf{W}_3^{'} & \mathbf{W}_4^{'} \\
\end{array}
\right),
\end{equation}
where $\mathbf{W}_3^{'}$ is the submatrix of zeros, the entries of $\mathbf{W}_1^{'}$ are
\begin{equation}
w^{'}_{1,ij} =\Gamma(i+j-1), \quad 1 \leq i, j \leq l-1,
\end{equation}
and the entries of $\mathbf{W}_4^{'}$ are
\begin{equation}
\begin{split}
w^{'}_{4,ij} &= \left(\frac{2}{\eta}-\vert N_t-N_r\vert-j-l\right)_{l-1}\Gamma(\frac{2}{\eta}-\vert N_t-N_r\vert-l+i-j), \\
&\quad \quad \quad l \leq i,j \leq N_{\min}-1.
\end{split}
\end{equation}

\begin{equation}
\vert \mathbf{W} \vert = \vert \mathbf{W}^{'}_1 \vert \vert \mathbf{W}^{'}_4 \vert
\end{equation}

By Lemma \ref{l3},
\begin{equation}
\vert \mathbf{W}^{'}_1 \vert = \prod_{k=1}^{l-1}\Gamma(k)\Gamma(k+1).
\end{equation}

By Lemma \ref{l8},
\begin{equation}
\vert \mathbf{W}^{'}_4 \vert = (-1)^{n_2} \prod_{j=l}^{N_{\min}-1}\left(\frac{2}{\eta}-\vert N_t-N_r\vert-j-l\right)_{l-1} \prod_{k=1}^{N_{\min}-l}\Gamma(k)\Gamma(\frac{2}{\eta}-N_{\max}+k).
\end{equation}
where $n_2= \frac{(N_{\min}-l)(N_{\min}-l-1)}{2}$.

Consequently, in terms of Theorem \ref{l6},
\begin{equation}
\begin{split}
\lim_{\bm{\Sigma}\rightarrow\mathbf{I}}\mu^*_{\mathrm{cor}} & = (-1)^{n_1+n_2+n_3} \prod_{k=1}^{N_{\min}-l}\frac{\Gamma(\frac{2}{\eta}-N_{\max}+k)} {\Gamma(\frac{2}{\eta}-\vert N_t-N_r\vert-k-2l+1)} \\
&\quad \times \frac{\Gamma(\vert N_t-N_r\vert-\frac{2}{\eta}+l+k)} {\Gamma(\vert N_t-N_r\vert - \frac{2}{\eta}+2l+k)}\; \mu^*_{\mathrm{unc}}
\end{split}
\end{equation}
where $n_3 = \frac{l(l-1)}{2}$.
Since for any function $f(x)$,
\begin{equation}
\prod_{k=1}^{N_{\min}-l}f(a+N_{\min}-k-l+1) = \prod_{k^{'}=1}^{N_{\min}-l}f(a+k^{'})
\end{equation}
where $k^{'} = N_{\min}-k-l+1$,
\begin{equation}
\begin{split}
\lim_{\bm{\Sigma}\rightarrow\mathbf{I}}\mu^*_{\mathrm{cor}}(\eta) &= (-1)^{n_1+n_2+n_3} \prod_{k=1}^{N_{\min}-l}
\frac{(\frac{2}{\eta}-N_{\max}+k-l)_l}{(N_{\max}-\frac{2}{\eta}-k+1)_l}\; \mu^*_{\mathrm{unc}}(\eta).
\end{split}
\end{equation}
By Lemma \ref{l8},
\begin{equation}
\left(\frac{2}{\eta}-N_{\max}+k-l\right)_l  = (-1)^l\left(N_{\max}-\frac{2}{\eta}-k+1\right)_l
\end{equation}
Thus,
\begin{equation}
\lim_{\bm{\Sigma}\rightarrow\mathbf{I}}\mu^*_{\mathrm{cor}}(\eta) = (-1)^{n_1+n_2+n_3+n_4}\,\mu^*_{\mathrm{unc}}(\eta).
\end{equation}
where $n_4 = l(N_{\min}-l+1)$.
As
\begin{equation}
(-1)^{n_1+n_2+n_3+n_4} = (-1)^{n_1-n_2+n_3+n_4} = 1,
\end{equation} we have
\begin{equation}
\lim_{\bm{\Sigma}\rightarrow\mathbf{I}}\mu^*_{\mathrm{cor}}(\eta) = \mu^*_{\mathrm{unc}}(\eta).
\end{equation}
This concludes the proof.


\end{document}